\documentclass[floats,reqno,oneside]{amsart}
\usepackage{bm,amsmath,amsfonts,amssymb,amsthm,stmaryrd}
\usepackage{graphicx}
\usepackage{hyperref}
\usepackage{color}
\usepackage[top=1in, left=1.5in, right=1in, bottom=1.25in]{geometry}
\usepackage{amsaddr}

\numberwithin{equation}{section}
\newcommand{\beq}{\begin{equation}}
\newcommand{\eeq}{\end{equation}}
\newcommand{\beqa}{\begin{eqnarray}}
\newcommand{\eeqa}{\end{eqnarray}}
\newcommand{\beqan}{\begin{eqnarray*}}
\newcommand{\eenan}{\end{eqnarray*}}

\def\Bbb{\mathbb}

\newcommand{\Dslash}{{\slash{\kern -0.5em}\partial}}
\newcommand{\Aslash}{{\slash{\kern -0.5em}A}}

\def\sqr#1#2{{\vcenter{\hrule height.#2pt
     \hbox{\vrule width.#2pt height#1pt \kern#1pt
        \vrule width.#2pt}
     \hrule height.#2pt}}}

\def\square{\mathchoice\sqr68\sqr68\sqr{4.2}6\sqr{3.0}6}
 
\def\thinspace{\kern .16667em}

\def\xp{x_{{\kern -.2em}_\perp}}
\def\subp{_{{\kern -.2em}_\perp}}

\def\defeq{:{\kern -0.5em}=}

\newcommand{\dom}{\mathrm{dom}\,}

\newcommand{\epi}{\mathrm{epi}\,}
\newcommand{\sub}{\mathrm{sub}\,}
\newcommand{\co}{\mathrm{co}\,}
\newcommand{\cco}{\overline{\mathrm{co}}\,}

\newcommand{\pair}[2]{\langle #1 , #2 \rangle}

\newcommand{\Trnorm}[1]{ \|{#1}\|_{tr} }

\newcommand{\Hilb}{{\mathcal H}}

\newcommand{\DiracProj}[1]{|#1\rangle\langle #1 |}

\newcommand{\braket}[2]{\langle #1 | #2 \rangle}

\newcommand{\Bool}{\mathsf{Bool}}
\newcommand{\Real}{{\mathbb R}}
\newcommand{\Nat}{{\mathbb N}}

\newcommand{\Numparts}{\mathcal N}
\newcommand{\JN}{{\mathcal J}_\Numparts}
\newcommand\dens{\mathrm{dens}\,}

\newcommand\EE{{\mathcal E}}
\newcommand\States{\mathsf{State}}
\newcommand\PureStates{\mathsf{State}^\circ}
\newcommand\Dens{\mathsf{Dens}}
\newcommand\DensN{\Dens_{\Numparts}}
\newcommand{\Fpure}{F^\circ}
\newcommand{\Potls}{\mathsf{Potl}}
\newcommand{\PotlsUp}{\Potls^\uparrow}
\newcommand{\PotlsBdd}{\mathsf{Bdd}}
\newcommand{\PotlsSimple}{\Potls^f(\Part)}

\newcommand{\PotlsF}{\mathsf{FSmall}}
\newcommand{\Lieb}{\mathsf{Lieb}}
\newcommand{\Const}{\mathsf{Const}}

\def\Part{{\mathfrak P}}

\newcommand\Star{{{}^*}}
\newcommand{\Std}[1]{{}^\circ{#1}}
\newcommand{\st}{\mathrm{st}\, }
\newcommand{\near}{\simeq}
\newcommand{\Univ}{{\mathbb U}}

\theoremstyle{plain}
\newtheorem{lem}{Lemma}[section]
\newtheorem{thm}{Theorem}[section]
\newtheorem{cor}{Corollary}[section]
\newtheorem*{cor*}{Corollary}
\newtheorem*{thm*}{Theorem}
\newtheorem{prop}{Proposition}[section]

\theoremstyle{definition}
\newtheorem{defn}{Definition}[section]

\theoremstyle{remark}
\newtheorem{rem}{Remark}[section]

\begin{document}

\title{On $F$ and $E$, in DFT}
\author{Paul~E.~Lammert}
\email{pel1@psu.edu}
\address{Dept. of Physics, 104B Davey Lab \\ 
Pennsylvania State University \\ University Park, PA 16802-6300}
\date{\today}

\begin{abstract}
Rigorous mathematical foundations of density functional theory are revisited,
with some use of infinitesimal (nonstandard) methods.
A thorough treatment is given of basic properties of internal
energy and ground-state energy functionals along with several improvements 
and clarifications of known results. A simple metrizable topology is 
constructed on the space of densities using a hierarchy of spatial partitions.
This topology is very weak, but supplemented by control of internal energy, 
it is, in a rough sense, essentially as strong as $L^1$.
Consequently, the internal energy functional $F$ is lower semicontinuous 
with respect to it.
With separation of positive and negative parts of external potentials, 
very badly behaved, even infinite, positive parts can be handled. 
Confining potentials are thereby incorporated directly into the density 
functional framework.
\end{abstract}
\maketitle

\tableofcontents

\section{Introduction}
\label{sec:intro}

\subsection{Motivations}
\label{sec:motivation}

In the density functional theory (DFT) literature, whenever the matter of 
mathematically rigorous foundations arises, a 1983 paper by 
Elliott Lieb\cite{Lieb83} justly looms large. 
It propounds what could reasonably be called the ``standard framework''.
Although (or perhaps because) the community seems generally to regard it 
as a satisfactory foundation, the literature gives evidence
that it is not widely understood.
Standard DFT textbooks\cite{Parr+Yang,Dreizler+Gross,Eschrig,Martin}
and reviews\cite{Capelle06} spill very little ink on such matters
since they have much else to cover.
What few expository treatments exist\cite{vanLeeuwen03,Eschrig} hew
very close to Lieb, and are not easy to obtain.
Consequently, one more such exposition may have a place. 
However, while these notes are a fairly complete
treatment of the most fundamental matters, it does not just add another
set of footprints to the same patch of ground. 
The standard framework locates densities in $L^1({\Real}^3)\cap L^3({\Real}^3)$
and external one-body potentials in the dual space 
$L^\infty({\Real}^3)+L^{3/2}({\Real}^3)$, with the internal energy function 
$F$ on the former and the ground-state energy $E$ on the latter in a relation 
of Fenchel conjugacy.
Rigid adherence to this framework is both unnecessarily
limiting and physically distorting. Consider first the density side.
In its primitive conception, density is a {\em measure}, telling us
the mass in every region. The natural topologies for measures are weak 
topologies. This is reinforced by the fact (\S \ref{sec:F-lsc}, apparently 
not noticed before, that $F$ is lower semicontinuous with respect to
topologies weaker than weak-$L^1$. 
Keeping in mind that densities are non-negative and normalized,
not general elements of a {\em vector} space is important
to seeing that. On the potential side, the standard framework is 
even more constricting. It is good physics, and may even be useful,
to allow potentials to have vastly more ill-behaved
positive parts than negative parts.
Why is $E(v) > -\infty$ for every $v\in L^{\infty}+L^{3/2}$? 
The general theory of Fenchel conjugacy, by itself,
certainly does not make that expected. 
A very clear understanding (Thm. \ref{thm:stability})
is reached by maintaining a clear view of what is required
to keep positive or negative unboundedness under control.

\subsection{What is here, old and new}
\label{sec:what-is-here}

This section gives a brief tour of the contents of the paper, 
emphasizing the new results and putting it in context
with some initial casual motivation.

It all begins with the following problem of conventional Hilbert-space
quantum mechanics.
A system of $\Numparts$ indistinguishable particles (electrons, if you wish)
with kinetic energy (operator) $T$ and mutual interaction $W$ is placed in 
an external one-body potential $v$, and we want to find the ground state and 
its energy. That is, a minimizer of 
$\langle \psi | T + W + \underline{v} | \psi \rangle$
is sought:
\begin{equation}
E(v) = \inf_\psi
\left\{ \langle \psi| T+W | \psi\rangle + \langle \psi | \underline{v} | \psi \rangle \right\}.
\nonumber
\end{equation}
$\underline{v}$ here is the $\Numparts$-body potential obtained by adding
a copy of $v$ for each particle in the system. But building $\underline{v}$
is a step in the wrong direction.
The external potential energy depends on the state only through its one-particle
density. So, instead of minimizing over all states in one go, why not split
it into two stages? First, for each density $\rho$, find the state with lowest
internal energy,
\begin{equation}
\Fpure(\rho) = \inf_{ \dens \psi = \rho} \langle \psi | T + W | \psi \rangle.
\label{eq:Fpure-min}
\end{equation}
Since all states with density $\rho$ have the same external energy, only
the best among them has any chance of being a ground state. Then, in
stage two, minimize over densities:
\begin{equation}
E(v) = \inf_{\rho} \left\{ \Fpure(\rho) + \int v \rho\right\}.
\label{eq:E-min}
\end{equation}
This is the constrained-search approach of Levy\cite{Levy82} and Lieb\cite{Lieb83}
The original problem involved particular $\Numparts$, $W$ and $v$. 
This final formula suggests an entirely different viewpoint.
The {\em system} specified by $\Numparts$ and $W$ is fixed, and we
study its {\em response} to {\em many different} external potentials.
All the relevant properties of the system are encoded in $\Fpure(\rho)$.

This may remind the reader of statistical mechanics. In a prototypical
magnetic system, there are an enormous number of internal degrees of
freedom, but the only thing relevant for the interaction with an
externally imposed magnetic field is the net magnetization.
The free energy as a function of magnetization is analogous to
$\Fpure(\rho)$, and the free energy as function of external field,
to $E(v)$. The two are related by a Legendre transformation,
similarly to (\ref{eq:E-min}.
But thermodynamic conjugacy works both ways. Do we also have
\begin{equation}
\Fpure(\rho) = \sup_v \left\{ E(v) - \int v \rho\right\} \, ?
\label{eq:F-max}
\end{equation}

Questions now come hard and fast; they are the subject of this paper. 
We have lots of infima here. Are there, in fact, minimizers?
Consider the functions $f(x) = 1/x$ and $g(x) =
(\text{ if }\, x = 0 \,\, \text{ then }\,\, 1 \,\, \text{ else }\,\, x^2)$.
Then, $\inf_x f(x) = \inf_x g(x) = 0$, but neither function actually
takes the value $0$ anywhere. 
Take minimizing sequences $f(x_n) \stackrel{n\to\infty}{\to} 0$, and 
$g(y_n)\stackrel{n\to\infty}{\to} 0$.
We would hope that at least subsequences converge to minimizers.
But it does not work. There is no subsequence of $(x_n)$ that converges
at all. The sequence $(y_n)$ converges to zero, but the function $g$ suddenly 
jumps up. In the case of $f$, there is
a failure of {\em compactness}, in the case of $g$, of {\em lower-semicontinuity}.
(If $g$ suddenly jumped {\em down} at zero, all would be well, so only
``half'' of continuity is really needed.)
In the context of function spaces there are many distinct, legitimate
and useful concepts of convergence.
When we say, ``$F$ is lower semicontinuous'', we want to make the
strongest possible statement, so we seek the {\em weakest topology}
with respect to which it is true. Lower-semicontinuity is a pervasive
theme in the following pages.

In \S \ref{sec:internal-e-as-fn-of-state}, we study the internal
energy as a function of pure and mixed states, and its
lower-semicontinuity.
\S \ref{sec:density} turns attention to density, where the
important results involve continuity of the map from states to densities
and compactness of the set of low-energy states with given density.
Pure state ($\Fpure$) and mixed-state ($F$) internal energy as a function 
of density is studied in \S \ref{sec:F}.
Several results here are new. Lower-semicontinuity with respect to
topologies much weaker than the $L^1$ topology is demonstrated.
A general method of transporting densities and wavefunctions is used to make
the well-known demonstration of an upper bound on $\Fpure$
in terms of the $H^1$ Sobolev norm of $\sqrt{\rho}$ less {\it ad hoc}
and to facilitate the calculation of a kinetic energy bound.
We point out that not only do finite-energy densities satisfy 
$\sqrt{\rho} \in H^1({\Real^3})$ as is well-known, 
but also $\rho \in W^{1,1}(\Real^3)$. 
Section \ref{sec:wk-P} presents the construction of a convenient very weak
yet metrizable topology on the space of densities, based on a hierarchy of
partitions of space into cubical lattices.
$F$ is lower semicontinuous with respect to this weak-$\Part$ topology.
Despite its weakness, the weak-$\Part$ topology is, roughly speaking,
essentially as strong as $L^1$ under conditions of control on the
internal energy.
Finally, in \S \ref{sec:energy-floor}, attention is turned to the
ground-state energy $E(v)$. Here we make a very strong, and uncustomary,
distinction between unbounded positive and unbounded negative potentials.
The physical motivation for this was touched upon in \ref{sec:motivation}.
Rewards for separating them in the mathematical treatment are the
possibility of incorporating confining potentials directly into the
density-functional framework, as well as an improved understanding
of the stability of Lieb's $L^\infty + L^{3/2}$ potentials, and continuity
of $E(v)$. We will see that Eq. (\ref{eq:F-max}) is not true, since the
right-hand side is necessarily convex, but $\Fpure$ is not.
The defect is easily corrected by replacing $\Fpure$ by the mixed-state internal
energy $F$. Furthermore, the maximization need be carried 
out only over potentials which are linear combinations of indicator 
functions of cells in the partition hierarchy $\Part$. This is a new result.
Appendix \ref{app:top} gives a refresher on basic functional analysis.
In these notes, we freely use simple methods of infinitesimal 
(or nonstandard) analysis. Infinitesimal methods are a way of doing mathematics, 
and especially of dealing with mathematical idealizations, that ought to be 
familiar to physical scientists. Appendix \ref{sec:nsa} gives an account
of what we need. 
It is relegated to an appendix simply because it is not the real subject of 
these notes, but only a tool. The reader may want to take an early look,
since it begins to be used already in \ref{sec:lsc-def}.

\subsection{Particle types}
\label{sec:particle-types}

We deal throughout with a system of $\Numparts$ indistinguishable particles.
A condensed matter physicist or chemist generally thinks about DFT
only as a theory of electrons.
Our default assumption is indeed that the particles are fermions,
but in principle, the theory could be applied to many different kinds of 
particles.
Statistics only appear in the explicit construction of a wavefunction
in \ref{sec:wavefunction-construction}, so for bosons a better bound
on $\Fpure$ would be achievable.
Spin is nothing but an annoyance here; it is always summed over
and makes no qualitative difference. 
Finally, although the mutual interaction of most interest is
the Coulomb interaction, the crucial ingredient is 
relative form-boundedness by kinetic energy. 
We do use positivity of the interaction in the form $F \ge 0$, 
but that is a matter of convenience and could be relaxed.

\section{Internal energy as function of state (${\mathcal E}$) }
\label{sec:internal-e-as-fn-of-state}

\subsection{Kinetic energy}
\label{sec:KE}
Throughout, the Hilbert space of $\Numparts$-particle wavefunctions
of appropriate symmetry is denoted by $\Hilb$.
In units such that $\hbar^2/2m$ has value 1,
the kinetic energy of the wavefunction $\psi(\underline{x},\underline{\sigma})$ is
\begin{equation}
\EE_0(\psi) \defeq \langle\psi | T | \psi\rangle 
= \int |\nabla\psi|^2
= \|\nabla\psi\|_2^2
\label{kinetic-energy}
\end{equation}
\begin{rem}[integration and configuration variable conventions]
$\underline{x}$ denotes the full set of $\Numparts$-particle configuration
space coordinates $(x_1,x_2,\ldots,x_{\Numparts})$, and similarly
$\underline{\sigma}$, the collection of spin coordinates.
The integral in (\ref{kinetic-energy}) illustrates some conventions
that we will use. Integrals are over all available variables unless
a restriction is indicated by a subscript on the integral sign.
For example, `$\int_{x_1=x} \cdots$' means that the position coordinate
of particle 1 is held fixed at $x$. 
The gradient symbol is subject to a {\em different} convention.
The gradient in (\ref{kinetic-energy}) is $3\Numparts$-dimensional;
a subscript would indicate that differentiation is with
respect to {\em only} the indicated coordinates.
Integration over spin variables means summation. 
\end{rem}
In case $\psi$ is in the domain of the operator $T$, so that
$T\psi \in L^2({\Bbb R})^3$ the quadratic form (\ref{kinetic-energy}) 
is equal to the inner product of $\psi$ with $T\psi$,
which would be written $\langle\psi|T\psi\rangle$.
The quadratic form, understood as a function 
$L^2({\Bbb R}^3) \rightarrow \overline{\Real}$
into the topped reals (see \ref{sec:lsc-def})
$\overline{\Real} \defeq {\Real}\cup\{+\infty\}$
has the advantage that it is meaningful for every $\psi$,
and there are many wavefunctions satisfying $\EE_0(\psi)<+\infty$
for which $T\psi$ does not exist.
In general, Hamiltonians appear in DFT in the guise of expectation
values. Quadratic form definitions are therefore appropriate.

\subsection{Interaction energy}
\label{sec:Coulomb}

Now we add the interaction energy of the particles.
\begin{equation}
\EE_\lambda(\psi) = \EE_0(\psi) + \langle\psi | W | \psi\rangle.
\end{equation}
Usually, the interaction $W$ is a Coulomb interaction, 
$\sum_{i<j} |x_i - x_j|^{-1}$, but we consider generally 
what properties we want to require.

Call $a > 0$ {\em an} $\EE_0$ form-bound for $W$ if there is some $b \ge 0$
such that
\begin{equation}
\forall \psi\in L^2({\Real}^3),\quad
|\langle\psi | W | \psi\rangle| \le a \EE_0(\psi) + b \|\psi\|^2.
\end{equation}
Then, if $a_0$ is the infimum of all $\EE_0$ form-bounds for $W$,
says that $W$ is $\EE_0$-form-bounded with relative bound $a_0$. 
The relative form bound of the Coulomb interaction is zero, as
shown in \ref{sec:Hardy}. 
For any $a>0$, there is some constant $b$ such that
$W_{Coul}$ is dominated by $\EE_0 + b$.
This fact is crucial to the stability of non-relativistic one-electron
atoms of arbitrary $Z$. The motif of dominating an attractive energy
with a repulsive one and its significance for stability will recur.
In the following, we assume that the interaction $W$ is non-negative
and $\EE_0$ bounded with relative form bound zero (``Kato tiny relative
to the kinetic energy''). Then,
for some constants $c,d > 0$,
\begin{equation}
d \|\nabla\psi\|_{2} \le E(\psi) \le c \|\psi\|_{H^1}.
\end{equation}

\subsubsection{Hardy inequality}
\label{sec:Hardy}

The proof that the Coulomb interaction is Kato tiny relative to 
kinetic energy uses the following inequality.
\begin{lem}[Hardy's inequality]
For $\psi \in H^1(\Real^3)$,
\begin{equation}
\int_{\Real^3} \frac{|\psi|^2}{4r^2}\, dx \le \int_{\Real^3} |\nabla\psi|^2\, dx
\nonumber 
\end{equation}
\end{lem}
\begin{proof}(sketch)
We demonstrate the inequality for {\it real}-valued $\psi \in C^\infty_c(\Real^3)$.
Real and imaginary parts separate, and an approximation argument extends it 
to $H^1(\Real^3)$.

Define $\phi = r^{1/2}\psi$, and note that $\phi(0) = 0$.
Now,
\begin{align}
(\nabla\psi)^2 
&= \left( -\frac{\nabla r}{r^{3/2}}\phi + r^{-1/2}\nabla\phi\right)^2
\ge \frac{1}{4r^3} \phi^2 - 
2\left(\frac{\phi}{2r^{3/2}}r^{-1/2}\right)\left(\frac{\partial\phi}{\partial r}\right)
\nonumber \\
&= \frac{1}{4r^2} \psi^2 - 
\frac{1}{2r^{2}}\frac{\partial\phi^2}{\partial r},
\nonumber
\end{align}
where in the inequality, we've thrown out a positive term.
The second term in the final expression integrates to zero,
since $\phi(x)=0$ at $x=0$ and for large $x$.
\end{proof}
A couple of tricks are needed to apply the Hardy inequality to our situation.
First, note that
$1/r$ differs from $f(r|\epsilon) = 
(\text{ if }\, r<\epsilon\, \text{ then }\, 1/r\, \text{ else }\, 0)$ 
by a bounded function and $f(r|\epsilon) \le \epsilon/r^2$.
Second, the expectation of $|x_i-x_j|^2$ can be handled by change of variables.
Write the required integral as
$\int |\psi|^2/|x_1-x_2|^2 \, d(x_1-x_2)\, dx_2\, dx_3\cdots dx_\Numparts$, and notice
that $|\nabla_{x_1-x_2}\psi|^2 \le |\nabla \psi|^2$.

\subsection{$\EE : \Hilb \rightarrow \overline{\Real}$ is lower semicontinuous}
\label{sec:E-lsc-pure}

\subsubsection{Lower semicontinuity and extended-real-valued functions}
\label{sec:lsc-def}

We will often be interested in functions taking values in the topped reals
\begin{equation}
\overline{\Real} \defeq \Real \cup \{+\infty\}.
\end{equation}
A neighborhood of $+\infty$ contains an interval of the form
$(a,\infty]$ for some $a\in\Real$.
Infinitesimally, this means that ``$x\near +\infty$'' is
to be understood literally.
$\Star{\overline{\Real}}$ consists of $\Star{\Real}$ together
with the new point $+\infty$. If $x\in \Star{\Real}$ is illimited,
then $\Std{x} = +\infty$.

Continuity of a function $f: X \rightarrow {\Real}$ at a point $x \in X$ 
is defined this way:
Given a tolerance $\epsilon$, there is a 
neighborhood $U$ of $x$ such that $|f(y) - f(x)| <  \epsilon$ for every $y\in U$. 
If we split the condition into $f(y) > f(x) - \epsilon$ and
$f(y) < f(x) + \epsilon$, then the first characterizes lower semicontinuity
at $x$, and the latter, upper semicontinuity.
In case $f$ is a map into $\overline{\Real}$, we extend this definition
for a point $x$ with $f(x) = +\infty$. $f$ is lower semicontinuous at $x$
if for any $M \in \Real$, there is a neighborhood
$U$ of $x$ such that $f(y) > M$ whenever $y\in U$.
In an infinitesimal idiom, $f: X\rightarrow\Real$ is lower semicontinuous
(abbreviated `lsc') at $x$ if $f(y) \gtrsim f(x)$ for all $y$ in the
halo of $x$. Equivalently, $\Std{f(y)} \ge f(x)$.
If $f$ is lsc everywhere, then we call it simply ``lower semicontinuous''.
This is the property we are usually interested in.
\begin{defn}
A $\Real$-valued function $f$ on a topological space $X$ is
lower semicontinuous if $f^{-1}(-\infty,a]$ is closed for every $a\in\Real$.
\end{defn}
We usually work with the infinitesimal characterization, but there
is also a useful geometrical intuition attached to the stated definition.
\begin{defn}
If $f: X \rightarrow \overline{\Real}$, then its effective domain is
\begin{equation}
\dom f \defeq \{x\in X : f(x) < \infty \}.
\label{eq:dom-def}
\end{equation}
The epigraph of $f$ is the set on or above its graph in
$X \times \Real$:
\begin{equation}
\epi f \defeq \left\{ (x,a) \in X\times\Real \, :\, f(x) \le a \right\}.
\label{eq:epigraph-def}
\end{equation}
\end{defn}
This implies that if $f(x) = \infty$, then there are no points in
$\epi f$ of the form $(x,a)$. Implicitly, $X\times\Real$ is equipped 
with the product topology, so that $(y,b) \near (x,a)$ if and only if
both $y\near x$ and $b\near a$.

We have introduced the notion of epigraph here because it 
provides another way to look at lower semicontinuity.
$f: X \rightarrow \overline{\Real}$ is lower semicontinuous if and only
if its epigraph is closed.
This perspective also gives a nice proof of the fact that
the pointwise supremum of any set of lower semicontinuous
functions is also lsc.
Let $g = \sup_{\alpha \in I} f_\alpha$, with all the $f_\alpha$'s lsc.
Then, $\epi g = \cap_\alpha \epi f_\alpha$.
But, the intersection of any collection of closed sets is closed,
hence $g$ is lsc.

\subsubsection{Lsc theorem}

\begin{thm}
\label{thm:E-lsc-on-H} 
$\EE : {\Hilb} \rightarrow \overline{\Real}$ is lower semicontinuous.
\end{thm}
\begin{proof}
Since $\EE \sim \EE_0$, it suffices to prove this for $\EE_0$.
And that is very easily done by Fourier transformation:
\begin{equation}
\EE_0(\psi) = \|\nabla\psi\|_2^2 = \int q^2 |\hat{\psi}(q)|^2\, dq
= \lim_{Q\to\infty} p_Q(\psi),
\nonumber
\end{equation}
where
\begin{equation}
p_Q(\psi) \defeq \int_{|q| \le Q} q^2 |\hat{\psi}(q)|^2 \, dq.
\nonumber
\end{equation}
But, $p_Q : \Hilb \rightarrow {\Real}$ is clearly continuous,
and $p_Q(\psi)$ is monotonically non-decreasing in $Q$.
As supremum of continuous functions, then, $\EE_0$ is lsc.
\end{proof}

\subsection{Extension to mixed states}
\label{sec:mixed-states}

\subsubsection{Standard and nearstandard mixed states}

To the vector ${\psi} \in \Hilb$ is associated the pure state 
$\DiracProj{\psi} \in \PureStates$, which is insensitive to the phase of $\psi$.
A mixed state is a probabalistic mixture of pure states, and
has the normal form
\begin{equation}
\gamma = \sum_{i\in{\mathcal I}} \DiracProj{\psi_i}, 
\quad j\not= k \Rightarrow \braket{\psi_j}{\psi_k}= 0,
\quad \|\psi_1\| \ge \|\psi_2\| \ge \cdots > 0,
\quad \sum_{i\in{\mathcal I}} \|\psi_i\|^2 < \infty.
\end{equation}
The index set ${\mathcal I}$ here is either $\Nat$, or $[1,\ldots,n]$ 
for some $n\in\Nat$.
Note, the normal form is not uniquely specified in case of degenerate eigenvalues.
On a few occasions, we have use for an alternate normal form
\begin{equation}
\gamma = \sum_{i\in{\mathcal I}} c_i \DiracProj{\hat{\psi_i}}.
\end{equation}
The two are related by $c_i = \|\psi_i\|^2$, so that the $\hat{\psi_i}$'s
are unit vectors.

We can also think of the set of mixed states $\States$ as a convex cone in the 
Banach space of trace-class operators equipped with the trace norm,
\begin{equation}
\|A\|_{\mathrm{Tr}} = \sum \braket{\phi_i}{|A|\phi_i}, \,\, \{\phi_i\} \, \text{ an ONB}.
\end{equation}
For a positive trace-class operator, the trace norm is just the trace,
and for a mixed state this is the generalization of norm-squared for
a state vector. 

For ${\psi}\in\Star{\Hilb}$, $\Std{\psi}$ is the standard vector infinitely close 
(``near'') to it, if such exists, else it is undefined. The map 
$\psi \mapsto \DiracProj{\psi}$ from $\Hilb$ to $\PureStates$ ought to be
continuous, so $\st \DiracProj{\psi} = \DiracProj{\Std{\psi}}$.
Continuing in this vein, to deduce the normal form of $\st \gamma$
without explicitly working with the trace norm, note that both
addition in $\States$ and $\mathrm{Tr}$ should be continuous.
Combined with the standard part operation in $\PureStates$, we conclude that
\begin{equation}
\st \gamma = 
\,\,
\text{ if } 
\Std{(\sum_{{\mathcal I}} \|\psi_i\|^2)} = \sum_{\Std{\mathcal I}} \|\Std{\psi_i}\|^2
\,\text{ then }\,
\sum_{i\in \Std{\mathcal I}} \DiracProj{\Std{\psi_i}}
\,\, \text{else }\mathsf{undefined}.
\label{eq:state-std-part}
\end{equation}
Thus, $\gamma$ is remote if any of the $\psi_i$'s is so, 
or if its norm is illimited,
or if an appreciable weight is carried by infinitesimal $\psi_i$'s,
or if the norm of $\gamma$ is illimited (the if-clause fails because
the first standard part is undefined).

\begin{defn}
The set of mixed states based on $\Hilb$ is denoted $\States$, 
and equipped with the trace norm.
The subset of pure states is denoted by $\PureStates$, and the
normalized states by $\States_1$. `State' with no modifier means
a member of $\States$, and the distinction between a pure state
and a corresponding state vector is often elided.
\end{defn}

\subsubsection{Lower semicontinuity of $\EE$}

With (\ref{eq:state-std-part}) and the nonstandard characterization of
lower-semicontinuity, we can now extend Thm. \ref{thm:E-lsc-on-H}
to $\States$. However, it will not be used in the following. 
\begin{cor}
\label{cor:E-lsc-on-mixed-states} 
$\EE: {\States} \rightarrow \overline{\Real}$ is lower semicontinuous.
\end{cor}
\begin{proof}
Take $\gamma \in {\States}$ a standard state and
and suppose 
$\gamma^\prime = \sum_{{\mathcal I}} \DiracProj{\psi_i^\prime}$
is {\it near} $\gamma$.
We need to show that $\EE(\gamma^\prime) \gtrsim \EE(\gamma)$.
Since $\EE$ is non-negative and lsc on $\PureStates$,
\begin{equation}
\Std{\EE(\gamma^\prime)} \ge \st \sum_{i=1}^n  \EE(\psi_i^\prime) 
= \sum_{i=1}^n  \Std{\EE(\psi_i^\prime)}
\ge \sum_{i=1}^m  \EE(\psi_i).
\nonumber
\end{equation}
If ${\mathcal I}$ is finite, just take $n = |{\mathcal I}|$, otherwise,
the limit $n\to\infty$ yields $\Std{\EE(\gamma^\prime)} \ge \EE(\gamma)$.
\end{proof}

\section{Density ($\rho$)}
\label{sec:density}

The one-particle density associated with an $\Numparts$-particle wavefunction 
$\psi$ is
\begin{equation}
(\dens \psi)(x) = 
\Numparts \int_{x_1=x} |\psi|^2\,.
\label{eq:density-def}
\end{equation}
The right-hand side here is a linear function of the rank-one operator
$\DiracProj{\psi}$ and thus has an immediate linear extension to 
mixed states:
\begin{equation}
\dens \left(\sum c_i \DiracProj{\psi_i} \right) = \sum c_i \dens \psi_i.
\end{equation}
We will overload the notation `$\dens$' by writing $\dens \psi$ 
for $\dens (\DiracProj{\psi})$.

\subsection{$\dens$ is continuous}
\label{sec:dens-cts}

Integrating the previous display over $x$, one finds
\begin{equation}
\|\dens \gamma \|_1 = \Numparts \Trnorm{\gamma}.
\label{eq:dens-bdd-linear-op}
\end{equation}
Densities are, of course, non-negative integrable functions.
We use the notations
\begin{defn}
\begin{align}
\Dens &= \left\{ \rho \in L^1(\Real^3): \rho \ge 0\right\},
\nonumber \\
\DensN &= \left\{ \rho \in \Dens: \int \rho\, dx = \Numparts \right\}.
\end{align}
\end{defn}

\begin{thm}
\label{thm:dns-is-cts}
$\dens : \States \rightarrow \Dens$ is continuous.
\end{thm}
\begin{proof}
(\ref{eq:dens-bdd-linear-op}) shows that $\dens$ is actually a
bounded linear map.
\end{proof}

\subsection{Tameness and nearstandardness}
\label{sec:tame-implies-ns}

\subsubsection{Tameness}
\label{sec:tameness}

A density in $\DensN$ always integrates to $\Numparts$. Thus, given
a density $\rho$, there is always some function (many in fact)
$R(\epsilon): (0,N] \rightarrow {\Real}^+$ such
that $\int_{|x| \ge R(\epsilon)} \rho\, dx < \epsilon$ for all $\epsilon$.
When a set $S$ of densities can all be controlled in this way by the 
same falloff function, $S$ would be said to be {\it tight}.
In a nonstandard context, we borrow this term to apply to a
single $\Star$density in $\Star{\DensN}$. By Leibniz' Principle,
$\rho \in \Star{\DensN}$ has a $*$-falloff function.
We say it is tight if $R(\epsilon)$ is limited for non-infinitesimal
$\epsilon$. Here is an alternative phrasing.
\begin{defn}
A $*$-density $\rho$ is {\it tight} if
$\int_{|x| \ge R} \rho\, dx \near 0$ for every illimited $R$.
A $*$-state $\gamma$ is {\it tame} if it has limited internal energy 
and norm, and
$\dens \gamma$ is tight. A $\Star$-state which is not tame is {\it wild}.
\end{defn}

\subsubsection{A fundamental spatial approximation}
\label{sec:in-a-box}

\newcommand{\cut}{\Lambda}
\newcommand{\BX}{\square}
Our approximation scheme has two ingredients, a spatial truncation method
and a family of finite-dimensional subspaces of $\Hilb$. Start with the former.
Let $0 \le \eta({x}) \le 1$ be a continuously differentiable cutoff function 
equal to one in the cubical box $\square(1/2) = [-1/2,1/2]^3$ and supported in
$\square({1}) = [-1,1]^3$, for example, a product of three 1D cutoff functions,
and for $R > 0$, let $(\eta_R)(x) = \eta(x/R)$.
Then, define the $\Numparts$-particle cutoff operator by
\begin{equation}
(\Lambda_R \psi)(\underline{x},\underline{\sigma}) =
\left( \prod_{i=1}^\Numparts \eta_R({x}_i) \right)\psi(\underline{x},\underline{\sigma}).
\end{equation}
Thus, the state $\cut_R \psi$ is supported in the $\Numparts$-particle
box $\BX^\Numparts(R)$, which is to say that the density is zero outside
$\BX(R)$. 
$\Lambda_R\psi$ agrees with $\psi$ inside $\BX^\Numparts(\frac{R}{2})$
and is an attenuated version outside, so we can get an upper bound
on the norm of the difference by just integrating $|\psi|^2$ over the
complement:
\begin{equation}
\|\Lambda_R\psi - \psi\|^2 
 \le \int_{\BX^\Numparts(\frac{R}{2})^c} |\psi|^2 \, 
\le \int_{\BX(\frac{R}{2})^c} \dens \psi
\label{ineq:cutoff-proximity}
\end{equation}
Bounding the energy of $\Lambda_R\psi$ is a little harder.
Using $\nabla(\Lambda_R \psi) = {\Lambda}_R \nabla \psi + (\nabla {\Lambda}_R) \psi$
and $|a+b|^2 \le 2|a|^2 + 2|b|^2$, we find
\begin{align}
\EE_0(\Lambda_{{R}}\psi) 
&\le 
\EE_0(\psi) 
+ \int_{\left[\BX^\Numparts(\frac{R}{2})\right]^c} |\nabla\psi|^2 \, d\underline{x}
+ 2\int_{\left[\BX^\Numparts(\frac{R}{2})\right]^c} 
|\nabla\Lambda_R|^2 |\psi|^2 \, d\underline{x}
\nonumber \\
&\le 
2\EE_0(\psi) 
+ 2\int_{{\BX(\frac{R}{2})}^c } |\nabla\eta_R|^2 (\dens\psi) 
\nonumber \\
&\le 
2\EE_0(\psi) 
+ \frac{\|\nabla\eta_1\|_\infty^2}{R^2}\int_{{\BX(\frac{R}{2})}^c } \dens\psi.
\label{ineq:cutoff-energy-bd}
\end{align}

To the box $\BX(R)$ we associate two closed subspaces of $\Hilb$.
Let $\{\varphi_1,\varphi_2,\ldots\}$ be the {\em single-particle}
particle-in-a-box eigenstates in $\BX(R)$.
From this single-particle basis, construct a complete 
product basis (adding spin and antisymmetrizing as needed)
for $\Hilb(\BX(R))$, the Hilbert space for $\Numparts$ particles
in the box.
$\Hilb_R$ is the subspace spanned by basis vectors with kinetic 
energy not exceeding $R$; it is a {\em finite-dimensional} Hilbert
space (asymptotically, $\mathrm{dim} \Hilb_R$ does not grow any faster
than $R^{9\Numparts/2}$).
Orthogonal projection $\Hilb(\BX(R)) \rightarrow \Hilb_R$ is denoted by $\pi_R$.
Now, if $\psi$ has density supported in $\BX(R)$,
then $\psi \in \Hilb(\BX(R))$, and if, in addition, 
$\EE_0(\psi) < \epsilon R$, then $\|\psi - \pi_R\psi\|^2 < \epsilon$.
We use the notation $\mathrm{Reg}_R$ to denote the composition of smooth truncation 
followed by projection:
\begin{equation}
\mathrm{Reg}_R \defeq \pi_R \circ \Lambda_R. 
\end{equation}

Now, if $\psi$ is tight and $R$ is illimited,
$\Lambda_R\psi \near \psi$ is immediate.
The inequality (\ref{ineq:cutoff-energy-bd}) shows that also
$\EE_0(\Lambda_R\psi) \lesssim \EE_0(\psi)$.

\subsubsection{Tame implies nearstandard}

Now we put the constructions of \ref{sec:in-a-box} to work.
Two observations are critical. The constructions were described 
for $R \in \Real_{>0}$, but work just as well for $R\in\Star{\Real}$. 
Also, since $\Hilb_R$ is finite-dimensional, with some finite
orthonormal basis $\{\varphi_1,\ldots,\varphi_n\}$, the same is
true of $\Star{\Hilb}_R$. The only difference is that $\phi \in \Star{\Hilb}_R$
has hypercomplex coefficients and may have illimited norm. But,
if $\phi = \sum \alpha_i \Star{\varphi_i}$ is limited, it is nearstandard
with standard part $\sum \Std{\alpha_i} \varphi_i$.

The translation into standard terms of
an external assertion to the effect that a property implies nearstandardness
is a claim about relative compactness. For instance, the following proposition 
is essentially a variation of the Rellich-Kondrashov theorem. We will not,
however, explicitly use the notion of compactness in these notes.
\begin{prop}
\label{prop:tame-implies-ns}
If $\psi \in \Star{\PureStates}$ is tame, then $\psi$ is nearstandard.
Furthermore, in that case, $\dens \st \psi = \st \dens \psi$
and $\EE(\st\psi) \le \st \EE(\psi)$.
\end{prop}
\begin{proof}

$\Hilb$ is a {\em complete} metric space, hence it suffices to 
show that $\psi$ can be approximated to any standard accuracy by a 
(standard) vector in $\Hilb$.

Since $\psi$ is tight, $\Lambda_R\psi \near \psi$ for
illimited $R$. Furthermore, according to inequality 
(\ref{ineq:cutoff-energy-bd}), $\Lambda_R\psi$ 
also has limited internal energy. Thus, the projection
$\pi_R$ will affect it only infinitesimally:
$\mathrm{Reg}_R \psi \near \psi$.

Now let $\epsilon > 0$ be given. The set
\begin{equation}
S(\epsilon) = 
\left\{ a \in \Star{\Real} : \|\mathrm{Reg}_a \psi - \psi \| < \epsilon \right\}
\nonumber 
\end{equation}
is internal, but it contains all illimited hyperreals. Thus, it must contain some
limited $R$. $\mathrm{Reg}_R \psi$ is a limited vector in $\Star{\Hilb_R}$,
and therefore nearstandard by the remarks immediately preceding the 
proposition. Thus, $\psi$ is within $2\epsilon$ of a standard vector.
But $\epsilon$ was arbitrary.

The last two statements of the Proposition now follow by
Thms. \ref{thm:E-lsc-on-H} and \ref{thm:dns-is-cts}.
\end{proof}

Following the familiar pattern, we now extend Prop. \ref{prop:tame-implies-ns}
to mixed states. This is not trivial and will involve one of the
trickiest pieces of reasoning in these notes.
\begin{cor}
\label{cor:tame-implies-ns-mixed-states}
If $\gamma \in \States_1$ is tame, then
$\gamma$ is nearstandard. Consequently,
$\dens \st \gamma = \st \dens \gamma$ and
$\EE(\st \gamma) \le \Std{\EE(\gamma)}$.
\end{cor}
\begin{proof}
Tackle this in contrapositive form. There are three ways $\gamma$ 
can be remote. If $\|\gamma\|_{\mathrm{Tr}}$ is illimited, $\gamma$ is
certainly wild, and if $\|\psi_i\|^2$ is appreciable for some 
remote $\psi_i$ in the support of $\gamma$, then $\gamma$ is wild
because Prop. \ref{prop:tame-implies-ns} says that $\psi_i$ is.
The third possibility is that $\gamma$ has a tail
$\gamma^\prime = \sum_{i\ge N} c_i \DiracProj{\hat{\psi}_i}$ with
$c_i \near 0$ for $i\ge N$ and $\sum_{i\ge N} c_i = M \gg 0$.
The following Lemma \ref{lem:wild-mixed-state} shows 
that $\gamma^\prime/M$ is wild. Hence, so is $\gamma$.
\end{proof}

\begin{lem}
\label{lem:wild-mixed-state}
Suppose the mixed state
\begin{equation}
\gamma = \sum c_i \DiracProj{\hat{\psi}_i} \in \Star{\States}_1
\nonumber
\end{equation}
has all coefficients infinitesimal: $c_i \near 0$.
Then, $\gamma$ is wild.
\end{lem}
\begin{proof}
Assume to the contrary that $\gamma$ is tame, so that
$\EE(\gamma) = E \ll \infty$ and  $\rho \defeq \dens \gamma$ is tight.
Define the {\it weight} of an internal set
$B \subset \Star{\Nat}$ to be $\mathrm{wt}\, B = \sum_{i\in B}c_i$,
Since the $c_i$ are infinitesimal, if $\mathrm{wt}\, B \gg 0$, then
$|B|$ is illimited. 
We are going to derive a contradiction as follows. 
For some {\em limited} $r$, we obtain an {\em illimited} set of indices
$A(r)$ such that
$\|\mathrm{Reg}_{r}\hat{\psi}_i - \hat{\psi}_i\| < \epsilon$ for
$i\in A(r)$.
Then $\|\mathrm{Reg}_{r}\hat{\psi}_i\| > 1-\epsilon > 1/2$, and
since the $\hat{\psi}_i$'s are orthonormal, a quick calculation reveals that
$\braket{\mathrm{Reg}_{r}\hat{\psi}_i }{\mathrm{Reg}_{r}\hat{\psi}_j }
\le ({2+\epsilon})/\epsilon < 3\epsilon$.
Thus, the cosine of the angle between $\mathrm{Reg}_{r}\hat{\psi}_i$ and 
$\mathrm{Reg}_{r}\hat{\psi}_j$ is less than $12\epsilon$.
For small $\epsilon$, this condition is a restriction;
a set of directions in the {\em limited-dimensional} space $\Hilb_r$ 
satisfying it must be limited, but $A(r)$ is not. 
Now we proceed to produce the required $r$ and $A(r)$.
Let standard $\epsilon \in (0,1/12)$ be given,
and define for $R\in\Star{\Real}$,
the {\em internal} set
\begin{equation}
A(R) \defeq
\left\{ i : 
\int_{\BX(R/2)^c} \dens \hat{\psi}_i \le \frac{1}{\epsilon} \int_{\BX(R/2)^c} \rho
\,\text{ and }\,
\EE_0(\hat{\psi}_i) \le \frac{E}{\epsilon}
\right\}.
\nonumber
\end{equation}
For every $R$, $\mathrm{wt}\, A(R) \ge 1-\epsilon$, for if this failed, 
there would be too much mass outside $\BX(R/2)$ or $\EE(\gamma)$ would exceed $E$.
Consequently, $|A(R)| > (1-\epsilon)/c_1 \near \infty$, independently of $R$.

If $R$ is illimited, then by
(\ref{ineq:cutoff-proximity},\ref{ineq:cutoff-energy-bd})
and tightness of $\rho$, $\Lambda_R \hat{\psi}_i \near \hat{\psi}_i$ and
$\EE_0(\Lambda_R \hat{\psi}_i) \ll \infty$ for $i \in A(R)$.
And, this implies $\|\mathrm{Reg}_R \hat{\psi}_i - \hat{\psi}_i\| \near 0$.
As a result, the {\em internal} set
\begin{equation}
C \defeq 
\{R\in\Star{\Real}: \forall i \in A(R), \,
\|\mathrm{Reg}_R \hat{\psi}_i - \hat{\psi}_i\| < \epsilon \}
\nonumber
\end{equation}
contains all illimited hyperreals. 
By underflow, $C$ contains some $r \ll \infty$. $A(r)$ is the 
hyperfinite set discussed in the first paragraph.
\end{proof}

\section{Internal energy as function of density ($\Fpure$ and $F$)}
\label{sec:F}

\subsection{Constrained energy minimization}
\label{sec:constrained-min}

We now move away from talking directly, or at least primarily, about states,
instead addressing them through the mediation of density.
The discussion of the Levy-Lieb constrained search formulation in \ref{sec:what-is-here}
introduced the idea of minimizing internal energy over all states with given 
density. However, it is now clear that we want to consider two different internal energy
functions of density, depending upon whether all states or only pure states are considered.
Thus, we define
\begin{align}
\Fpure(\rho) &= \inf \left\{ \EE(\psi) : \psi \in \PureStates, \, \dens \psi = \rho \right\}
\nonumber \\
F(\rho) &= \inf \left\{ \EE(\gamma) : \gamma \in \States, \, \dens \gamma = \rho \right\}.
\end{align}
Important questions about these functions are these: For which densities they are
finite? Are there useful upper/lower bounds in terms of explicit functions of density?
Are $\Fpure$ and $F$ lower semicontinuous? Are the infima in their definitions actually
realized (are they actually minima)?
Using the results of \ref{sec:E-lsc-pure} and \ref{sec:mixed-states},
we can immediately settle the last question.
\begin{thm}
\label{thm:F-infimum-is-min}
The infima in the definitions of $\Fpure$ and $F$ are realized.
\end{thm}
\begin{proof}
Let standard $\rho$ be given and assume $\Fpure(\rho) < +\infty$, as otherwise there is 
nothing to prove.
Let $\psi \in \Star{\Hilb}$ be such that
$\EE(\psi) \near \Star{\Fpure}(\Star{\rho})$. Since $\psi$ is thus tame,
it is nearstandard by Thm. \ref{prop:tame-implies-ns} so $\st \psi$ exists,
and according to Thm. \ref{thm:E-lsc-on-H}, $\EE(\st \psi) \lesssim \EE(\psi) \near \Fpure(\rho)$.
Since $\Fpure(\rho)$ and $\EE(\st\psi)$ are both standard, they are equal.

The proof for $F$ is entirely analogous, just substituting 
Cor. \ref{cor:E-lsc-on-mixed-states} for Thm. \ref{thm:E-lsc-on-H},
and Cor. \ref{cor:tame-implies-ns-mixed-states} for Thm. \ref{thm:E-lsc-on-H}.
\end{proof}

\subsection{Sobolev norms and kinetic energy bounds}
\label{sec:lower-bds-F}

In this section, we derive some important lower bounds on $\Fpure$ and $F$.

\subsubsection{Basic inequality}

Starting from (\ref{eq:density-def}), we deduce
\begin{align}
\nabla\rho(x) 
&\le 2\Numparts \int_{x_1=x} |\psi^* \nabla_{1}\psi|^2\, 
\le 2\Numparts 
\left(\int_{x_1=x} |\psi|^2 \, \right)^{1/2}
\left(\int_{x_1=x}|\nabla_{1}\psi|^2 \, \right)^{1/2}
\nonumber \\
&= 2 \rho(x)^{1/2}
\left(\int_{x_1=x}|\nabla_1 \psi|^2 \, \right)^{1/2}.
\label{ineq:rho-psi}
\end{align}
The subscript on $\nabla$ indicates that differentiation is 
only with respect to $x_1$.

\subsubsection{Two bounds}

There are two useful things we can do with this result.
First, squaring, dividing through by $\rho$ and noting that 
$|\nabla\rho|^2/\rho = |\nabla\sqrt{\rho}|^2$,
\begin{equation}
\|\nabla\sqrt{\rho}\|_2^2 \le \frac{4}{\Numparts} \EE_0(\psi).
\label{ineq:H1-E0-bound-pure}
\end{equation}
To extend this to a mixed state 
$\gamma = \sum |\psi_i\rangle\langle\psi_i |$,
abbreviate $\rho = \dens \gamma$, $\rho_i = \dens \psi_i$ 
and use $\nabla \rho = 2 \sqrt{\rho}\,\, \nabla\sqrt{\rho}$ and
the Cauchy-Schwartz inequality to obtain
\begin{equation}
\left| \sqrt{\rho}\,\, \nabla\sqrt{\rho} \right|
= \left| \sum_i \sqrt{\rho_i}\,\,\nabla\sqrt{\rho_i} \right|
\le \left(\sum_i \rho_i\right)^{1/2} 
\left( \sum_i \left|\nabla\sqrt{\rho_i}\right|^2\right)^{1/2}.
\nonumber
\end{equation}
Since $\rho = \sum_i \rho_i$, divide through by $\sqrt{\rho}$ and
square to get
$\left|\nabla \sqrt{\rho} \right|^2 \le \sum_i \left|\nabla\sqrt{\rho_i}\right|^2$.
Integrating and inserting (\ref{ineq:H1-E0-bound-pure}) yields
\begin{equation}
\|\nabla(\dens\gamma)^{1/2}\|_2^2 \le \frac{4}{\Numparts} \EE_0(\gamma).
\label{ineq:H1-E0-bound}
\end{equation}
For given $\rho$, this applies to any $\gamma$ with $\dens \gamma = \rho$ to give
\begin{equation}
\|\nabla\sqrt{\rho}\|_2^2 \le \frac{4}{\Numparts} F(\rho).
\label{ineq:H1-F-lower-bd}
\end{equation}

Since we wish to work with density rather than the square root of
density, the inequality (\ref{ineq:H1-E0-bound}) is not always convenient.
Alternatively, returning to (\ref{ineq:rho-psi}),
integrating and applying the Cauchy-Schwartz inequality,
\begin{equation}
\int |\nabla\rho| \, 
\le 2 \left( \int \rho \, \right)^{1/2}
\left(\frac{1}{\Numparts}\int|\nabla\psi|^2 \, \right)^{1/2},
\nonumber
\end{equation}
which yields
$\|\nabla\rho\|_1 \le 2\sqrt{\EE_0(\psi)}$.
This is extended to mixed states by a method similar to what gave
(\ref{ineq:H1-E0-bound}), with the result
\begin{equation}
\|\nabla\dens\gamma\|_1 \le 2\sqrt{\EE_0(\gamma)}.
\label{ineq:W11-E0-bound}
\end{equation}
Again, we turn this into a bound for $F$, as
\begin{equation}
\|\nabla\rho\|_1 \le 2\sqrt{F(\rho)}.
\label{ineq:W11-F-lower-bd}
\end{equation}

\subsection{Effective domain}

In this section we will show that
\begin{equation}
\boxed{
\dom \Fpure = \dom F = 
\JN \defeq \left\{\rho \in \DensN \, :\,  \sqrt{\rho} \in H^1({\Bbb R})^3 \right\}
}
\label{eq:effective-domains}
\end{equation}
is the effective domain of both $\Fpure$ and $F$. 
That $\JN \supseteq \dom F \supseteq \dom \Fpure$ follows from
the results of \ref{sec:lower-bds-F} and the fact that
$F \le \Fpure$. So what needs to be shown is $\JN \subset \dom\Fpure$.
This is done by an explicit construction.

\subsubsection{Transporting densities and wavefunctions}
\label{sec:transporting}

Suppose a density $\rho_0$ is given on some domain $\Omega$
and we wish to pull back the mass in $\Omega$ along a bijection
$\varphi : {\Real}^3 \rightarrow \Omega$. That is, we want the
density $\rho$ on ${\Real}^3$ such that
\begin{equation}
\int_A \rho = \int_{\varphi(A)} \rho_0
\nonumber
\end{equation}
for every measurable $A \subset {\Real}^3$.
The solution to this problem is very familiar:
with $J(x) = |\partial\varphi/\partial x|$ denoting the Jacobian
determinant and abbreviating $y = \varphi(x)$,
\begin{equation}
\rho(x) = \rho_0(\varphi(x))\, J(x) = \rho_0(y) J(x).
\label{density-transport}
\end{equation}
What if, instead, we are given a wavefunction 
$\psi_0$ on $\Omega^\Numparts$ with $\dens \psi_0 = \rho_0$ and wish to find 
$\psi$ with $\dens \psi = \rho$?
A solution is nearly as easy. First extend the map $\varphi$ to
the $\Numparts$-particle configuration space as
\begin{equation}
\underline{y} = \underline{\varphi}(\underline{x}) = 
(\varphi(x_1), \varphi(x_2), \ldots \varphi(x_\Numparts)).
\nonumber
\end{equation}
Then, we can take
\begin{equation}
\psi(\underline{x}) = \psi_0(\underline{y}) 
\, \Big|\frac{\partial \underline{\varphi}}{\partial \underline{x}}\Big|^{\frac{1}{2}}
= \psi_0(\underline{y}) \prod_i J(x_i)^{1/2}.
\end{equation}
Naturally, one wants to bound $\|\psi\|_{H^1({\Real}^3)}$ in terms of 
$\|\psi_0\|_{H^1(\Omega)}$, but we shall not pursue that problem in general,
but move on to the special case that interests us.

\subsubsection{Wavefunction construction for $\rho \in \JN$}
\label{sec:wavefunction-construction}

We now construct\cite{Lieb83,Harriman81,Zumbach85}
a wavefunction $\psi$ with $\dens \psi$ equal to given $\rho$ in $\JN$,
and calculate a bound on its internal energy.
Using \ref{sec:transporting} as inspiration, we find a map
$\varphi: {\Real}^3 \rightarrow \Omega$ onto region $\Omega$ and a type of 
density $\rho_0$ on that region for which an appropriate Slater 
determinant can be easily found. We leave spin coordinates out of
consideration, so implicitly all $\sigma$ are supposed to take
the maximum value.

Take $\Omega = [0,1]\times{\Real}\times{\Real}$, and assume
$\rho_0$ has the property
\begin{equation}
\int \rho_0(y^1,y^2,y^3)\, dy^2\, dy^3 = {\Numparts}.
\label{eq:rho0-condition}
\end{equation}
Then, it is readily apparent that the one-particle wavefunctions
[$y \equiv (y^1,y^2,y^3)$]
\begin{equation}
\phi_{0,k}(y) \defeq \left(\frac{\rho_0(y)}{\Numparts}\right)^{1/2} e^{2\pi i k y^1}
\label{eq:phi-0-k}
\end{equation}
indexed by integer $k$ are orthonormal over $\Omega$.
Since the density for each $\phi_k$ is
$\rho_0(y)/\Numparts$, a Slater determinant of $\phi_k$'s for
$\Numparts$ distinct $k$'s will give density $\rho_0$.

All that is now required is a mapping of $\Real^3$ onto $\Omega$ which
carries the given $\rho$ on to a $\rho_0$ satisfying (\ref{eq:rho0-condition}).
With the definition
\begin{equation}
J(x^1) \defeq \frac{1}{\Numparts}\int \rho(x) \, dx^2\, dx^3,
\label{eq:J-def}
\end{equation}
one immediately checks that the following works.
\begin{align}
\varphi(x) &= (y^1,y^2,y^3) = (y^1,x^2,x^3),
\nonumber \\
y^1 &= \int_{-\infty}^{x^1}J(z^1)\, dz^1.
\end{align}
Mass is being redistributed only along the first direction, at a
varying rate so that (\ref{eq:rho0-condition}) is satisfied.

Of course, the notation was chosen because $J$ actually is the Jacobian of $\varphi$:
\begin{equation}
J(x^1) = \frac{dy^1}{dx^1} = \left| \frac{ \partial y}{\partial x} \right| = \frac{\rho(x)}{\rho_0(y)}.
\label{eq:jacobian}
\end{equation}
According to (\ref{density-transport}), the $\phi_{0,k}$ of (\ref{eq:phi-0-k}) are
now transported to $\Real^3$ as
\begin{equation}
\phi_k(x) = \phi_{0,k}(y) \sqrt{J(x)} = 
\sqrt{ \frac{\rho(x)}{\Numparts} } \,\, e^{2\pi i k y^1}.
\end{equation}
Without further calculation, the general construction assures us
that the $\phi_k$ are orthonormal on ${\Real}^3$ and that
a Slater determinant $\psi = |\phi_{k_1},\ldots,\phi_{k_\Numparts}|$ of them 
yields the desired density $\rho$.
As illustrated here, $y$ may be regarded as a function of $x$ or vice-versa, and
either one as a coordinatization of $\Real^3$ or of $\Omega$ as convenient,
because the correspondence between $x$ and $y$ is bijective.

The wavefunction now being constructed, all that remains is to
bound on the kinetic energy of $\psi$ by bounding the kinetic energy of 
each $\phi_k$:
\begin{equation}
\EE_0(\phi_k) = \frac{1}{\Numparts} \|\nabla\sqrt{\rho}\|_2^2  + 
\frac{(2\pi k)^2}{\Numparts} 
\int {\rho(x)} \left(\frac{d y^1}{d x^1}\right)^2 \, dx.
\label{eq:E0-phi-k}
\end{equation}
The second term here will occupy us with a bit of Cauchy-Scwhartz calisthenics.
Using (\ref{eq:jacobian}),
\begin{equation}
\int \left(\frac{d y^1}{d x^1}\right)^2 \rho(x) \, dx
= \int {\rho_0(y)} J^2 \, dy 
= \Numparts \int_0^1 J^2 \, dy^1.
\label{eq:E0-manips-1}
\end{equation}
Estimate the integrand in the final expression as
\begin{align}
J(x^1)^2 
&= \left( \int_{-\infty}^{x_1} {J(z)}^{\frac{1}{2}} 
\left( J(z)^{-\frac{1}{2}} \frac{d{J(z)}}{dz}\right) dz \right)^2
\le \left( \int J \, dx^1\right) 
\int {J}^{-1} \left|\frac{d{J}}{dx^1}\right|^2 dx^1
\nonumber \\
&= \int J^{-1} \left|\frac{d{J}}{dx^1}\right|^2 dx^1.
\label{ineq:manips-2}
\end{align}
The first equality is just the fundamental theorem of calculus, the inequality is
Cauchy-Schwartz, and the final equality follows from $\int J\, dx^1 = \int dy^1 = 1$.
To evaluate the final integral,
Again, we go to work on the integrand with the Cauchy-Schwartz inequality, to obtain
\begin{equation}
\left|\frac{dJ}{dx^1}\right|^2
= \left( \frac{1}{\Numparts}\int 2 \sqrt{\rho(x)} 
\left|\frac{\partial\sqrt{\rho(x)}}{\partial x^1}\right| \, dx^2\, dx^3 \right)^2
\le \frac{4}{\Numparts^2} {J(x^1)} 
\int |\nabla\sqrt{\rho}|^2 \, dx^2\, dx^3.
\nonumber
\end{equation}
Inserting this into (\ref{ineq:manips-2}),
\begin{equation}
J(x^1)^2 \le \int J^{-1} \left|\frac{d{J}}{dx^1}\right|^2 dx^1 
\le  \frac{4}{\Numparts^2} \|\nabla\sqrt{\rho}\|_2^2
\nonumber
\end{equation}
Finally, returning to (\ref{eq:E0-phi-k}) via (\ref{eq:E0-manips-1}) results in
\begin{equation}
\EE_0(\phi_k) \le \frac{1}{\Numparts}\left(1 + \left(\frac{4\pi k}{\Numparts}\right)^2\right)
\|\nabla\sqrt{\rho}\|_2^2.
\nonumber
\end{equation}
This is for one of the orbitals entering our Slater determinant, so
summing over $\Numparts$ values of $|k| < \Numparts$,
\begin{equation}
\Fpure(\rho) \le \EE_0(\psi) \le (1 + 16\pi^2) \|\nabla\sqrt{\rho}\|_2^2.
\label{ineq:final-explicit-construction}
\end{equation}

\subsection{$\JN \subset L^1 \cap L^3 \cap W^{1,1}$}

Combining the upper bound on $\Fpure(\rho)$ from (\ref{ineq:final-explicit-construction})
with the lower bound of (\ref{ineq:H1-F-lower-bd}) yields
\begin{equation}
\boxed{
c \|\sqrt{\rho}\|_{H^1}^2 \le F(\rho) \le \Fpure(\rho) \le c^\prime \|\sqrt{\rho}\|_{H^1}^2
}
\label{ineq:F-equiv-H1-sqrt-rho}
\end{equation}
for some $\Numparts$-dependent constants $c$ and $c^\prime$.
The demonstration of (\ref{eq:effective-domains}) is thus achieved,
with some extra information. 

Characterizations in terms of norms of $\rho$ rather than of $\sqrt{\rho}$
can also be useful. Inequality (\ref{ineq:W11-F-lower-bd}) shows that
$\JN \subset W^{1,1}$, where $W^{1,1}$ is the Sobolev space with norm
$\|f\|_{W^{1,1}} = \int (|f| + |\nabla f|)\, dx$.
Also\cite{Lieb83}, the Sobolev inequality
\begin{equation}
3 \left(\frac{\pi}{2}\right)^{4/3} \|f\|_6^2 \le
\|\nabla f\|_2^2, 
\nonumber
\end{equation}
with $\sqrt{\rho}$ substituted for $f$ yields
\begin{equation}
\|\rho\|_3 \le c F(\rho),
\label{L3-F-bnd}
\end{equation}
for some constant $c$.
Thus, $\JN \subset L^3$, as well.
The inequality (\ref{L3-F-bnd}) will be applied to the consideration
of potentials in \S \ref{Lieb-potls}.

\subsection{Internal energy is lower semicontinuous on $\DensN$}
\label{sec:internal-e-lsc}

\subsubsection{Topologies on $\DensN$}
\label{sec:dens-tops}

The subject of this section,
lower semicontinuity of $\Fpure$ and $F$ on $\DensN$,
is important to the Fenchel conjugacy with the 
ground-state energy $E$ discussed in \S \ref{sec:energy-floor},
among other things.
But, lower semicontinuous with respect to what topology on $\DensN$?
In \S \ref{sec:dens-cts}, we showed that $\dens : \States \rightarrow \Dens$
is continuous with respect to the $L^1$ topology on $\Dens$.
We have not made use of that result, yet.
Continuity of $\dens$ thus holds for any topology weaker than $L^1$.
That is, so long as the topology is not so weak that it fails to
be Hausdorff, $\dens \st \gamma = \st \dens \gamma$ for nearstandard $\gamma$.
(Recall that a Hausdorff topology is such that for $x\not= y$,
one can find neighborhoods $U$ of $x$ and $V$ of $y$, such 
that $U\cap V = \emptyset$.)
The topology is implicit in the standard part operation, which is not
well-defined for a non-Hausdorff topology. 
Continuity of $\dens$ is important for the proof that $\Fpure$ is lsc,
but there is another condition that we need to impose.
\begin{defn}
A topology $\tau$ on $\DensN$ is {\it weak-but-not-leaky} if
it is weaker than, or equivalent to, $L^1$,
yet Hausdorff and such that all nearstandard points are tight. 
\end{defn}
An obvious example of such a topology is the weak-$L^1$ topology
induced by the seminorms $\rho \mapsto \int f \rho$ for $f\in L^\infty$.
A disadvantage of the weak-$L^1$ topology is that it is not metrizable.
In \S \ref{sec:wk-P}, we will construct a metrizable weak-but-not-leaky
topology on $\DensN$, the weak-$\Part$ topology, which is even 
weaker than weak-$L^1$.

\subsubsection{Pure-state internal energy}
\label{sec:Fpure-lsc}

\begin{thm}
\label{thm:Fpure-lsc}
$\Fpure : (\DensN,\tau) \rightarrow \overline{\Real}$ is lsc,
if $\tau$ is weak-but-not-leaky.
\end{thm}
\begin{proof}
Let $\rho \in \DensN$ be given and $\rho^\prime \in \Star{\DensN}$
be $\tau$-near $\rho$. Assume $\Fpure(\rho^\prime) \ll \infty$,
since otherwise there is nothing to show. $\rho^\prime$ is tight
by definition of weak-but-not-leaky, so
there is a tame $\psi^\prime$ with $\dens \psi^\prime = \rho^\prime$,
and $\EE(\psi^\prime) \near \Fpure(\rho^\prime)$.
By Prop. \ref{prop:tame-implies-ns} and \ref{sec:dens-tops},
$\psi^\prime$ is nearstandard, with 
$\EE(\st\psi^\prime) \lesssim \EE(\psi^\prime) \near \Fpure(\rho^\prime)$,
and $\dens \st \psi^\prime = \st \dens \psi^\prime = \st \rho^\prime = \rho$.
\end{proof}

\subsubsection{Mixed-state internal energy}
\label{sec:F-lsc}

\begin{defn}
If $A$ is a subset of a vector space, then its
convex hull is defined as
\begin{equation}
\co A \defeq \left\{ \sum_{i=1}^n \alpha_i x_i \,:\, \sum_{i=1}^n \alpha_i = 1 
\,\,\text{ and }\,\, x_1,\ldots,x_n\in A \right\}.
\end{equation}
The closed convex hull of $A$, denoted $\cco A$ is the closure
of $\co A$ in whatever topology is being considered on the ambient 
vector space.
\end{defn}
The goal now is to extend the result of \ref{sec:Fpure-lsc} 
from $\Fpure$ to $F$ by showing that $\epi F = \cco \epi \Fpure$.
For, lower semicontinuity (convexity) of $F$ is equivalent to closure 
(resp., convexity) of its epigraph.
\begin{thm}
\label{thm:F-is-cco-of-Fpure}
$\epi F = \cco \epi \Fpure$ in $\DensN \times {\Real}$,
with any weak-but-not-leaky topology on $\DensN$.
\end{thm}
\begin{proof}
To show that $\epi F = \cco \epi \Fpure$, we establish
inclusion in both directions.

\noindent (a)
$\epi F \subseteq \cco \epi \Fpure$:
It suffices to show that, if $\dens \gamma = \rho$, 
then \hbox{$(\rho,\EE(\gamma)) \in \cco \epi \Fpure$}.
From the second normal form for $\gamma$, define
\begin{equation}
w_n = \sum_{i=1}^n c_i, \quad
\gamma_n = 
\sum_{i=1}^n \frac{c_i}{w_n} \DiracProj{\hat{\psi}_i}.
\nonumber
\end{equation}
Then,
\begin{equation}
\EE(\gamma_n) = 
\sum_{i=1}^n 
\frac{c_i}{w_n} 
\EE(\hat{\psi}_i) \to \EE(\gamma),
\quad
\dens \gamma_n = 
\sum_{i=1}^n \frac{c_i}{w_n} \dens \hat{\psi}_i \to \dens \gamma.
\nonumber
\end{equation}
$(\dens \gamma_n,\EE(\gamma_n))$ is in $\co \epi \Fpure$, 
because $(\dens \hat{\psi}_i,\EE(\hat{\psi}_i)) \in \epi \Fpure$.
\hfill\break
Therefore,
\hbox{$(\dens \gamma,\EE(\gamma)) \in \cco \epi \Fpure$,} 
as required.

\noindent (b) $\epi F \supseteq \cco \epi \Fpure$:
Let $(\rho,\lambda) \in \cco \epi \Fpure$ be given (standard).
If $\lambda = +\infty$, there is nothing to show, so assume
$\lambda < +\infty$.
By definition of closed convex envelope, there is $N \in \Star{\Nat}$
and collections of coefficients and orthonormal $\Star$-pure-states
such that
\begin{equation}
\sum_{i\le N} c_i = 1, \quad \rho \near \sum_{i\le N} c_i\, \dens(\psi_i), \quad
\lambda \near \sum_{i\le N} c_i\, \EE(\psi_i).
\end{equation}
But this shows that
$\gamma = \sum_{i=1}^N c_i |\psi_i\rangle\langle\psi_i | \in \Star{\States}_1$ 
satisfies
$\dens \gamma \near \rho$ and $\EE(\gamma) \near \lambda$.
Since $\gamma$ is thus tame, Cor. \ref{cor:tame-implies-ns-mixed-states}
implies that $\EE(\st{\gamma}) \le \lambda$ and
$(\dens \st{\gamma},\EE(\st{\gamma})) \in \epi F$.
\end{proof}
This of course immediately gives us lower-semicontinuity.
\begin{cor}
\label{cor:F-is-lsc}
$F: (\DensN,\tau) \rightarrow \overline{\Real}$ is lower semicontinuous 
if $\tau$ is weak-but-not-leaky.
\end{cor}

\section{Weak-$\Part$ topology on $\DensN$}
\label{sec:wk-P}

The aim in this section is to construct a convenient, metrizable,
weak-but-not-leaky topology on $\DensN$ which is called the 
weak-$\Part$ topology.
Lower semicontinuity of $F$ with respect to it immediately implies
the same for any $\|\cdot\|_p$, $1\le p < \infty$. 
Although we are interested in only $\DensN$, the weak-$\Part$
metric will be defined on all of $\Dens$.

\subsection{Hierarchy of partitions}

Let $\Part^0$ denote a partition of ${\Real}^3$ into a
regular grid of cubes of side length $1$, with corners
at integer coordinates. To make it a genuine partition,
we agree that each cell includes its front, left and bottom faces,
but not the back, right or top.
Now, divide each cell of $\Part^0$ into $2^3$ cubes of side length $2^{-1}$;
these latter cubes are the cells of $\Part^1$. 
Continue subdividing so that each cell of $\Part^n$ is the union of $2^3$ cells
of $\Part^{n+1}$ to obtain an infinite hierarchy of partitions
\begin{equation}
\Part^0 < \Part^1 < \cdots.
\nonumber
\end{equation}
A set $A$ is {\it $\Part^n$-measurable} if it is the union of cells of
$\Part^n$, {\it $\Part^n$-finitely-measurable} if it is a finite
such union, and simply {\it $\Part$-measurable} if a finite union of 
cells. In this last case, $A$ is actually $\Part^m$-finitely-measurable,
where $m$ is the maximal rank of members of $A$, but the terminology
allows us to avoid mentioning $m$.

We define a linear projection $\pi_n$ from integrable functions 
to $\Part^n$-measurable integrable functions as follows.
For a function $f \in L^1({\Real}^3)$,
$\pi_n f$ is constant over each cell of $\Part^n$, with value on 
a cell equal to the average of $f$ over that cell. 
Thus, $f$ and $\pi_n f$ have the same integral over each cell.

\subsubsection{a metric for the weak-$\Part$ topology}
\label{sec:metric-for-wk-P}

\begin{defn}
For $A$ a $\Part^n$-finitely-measurable set, 
the continuous seminorm $p_A$ on $L^1({\Real}^3)$ is defined as
\begin{equation}
p_A(f) \defeq \int_A |\pi_n f|\, \le \| f\|_1.
\label{wk-Part-seminorms}
\end{equation}
\end{defn}
The weak-$\Part$ topology is simply the topology generated by
these seminorms. However, not all these seminorms are needed.
The smaller collection $\{p_\Omega : \Omega \in \Part\}$ suffices.
That is, $f_n \stackrel{\mathrm{wk}-\Part}{\to} f$ if and only if
$p_\Omega(f-f_n) \to 0$ for every $\Omega \in \Part$.
We can go even further and use just $\{p_{C_n},n=0,1,\ldots\}$
where $C_n$ is the largest $\Part^n$-measurable set inside the
ball $B_0(n)$ of radius $n$ centered at $0$.
Finally, if a metric is preferred, 
\begin{equation}
d_\Part(f,g) = \sum_{n=0}^\infty \frac{p_{C_n}(f-g)}{2^{n+2}}
\label{d-Part}
\end{equation}
generates the weak-$\Part$ topology.
Note that for $f,g \in L^1$, $d_\Part(f,g) \le \|f-g\|_1/2$.

\subsection{How weak is it?}
\label{sec:how-weak}

An example shows an important way in which weak-$\Part$
is weaker than the $L^1$ topology.
Let $f$ equal $1$ on some cell $\Omega \in \Part^0$, and $0$
elsewhere. One can easily construct a function $f_n$ which is
zero off $\Omega$, and inside $\Omega$ equal to $2$ on half 
the $\Part^n$ cells and to $0$ on the other half in such a
way that $\pi_m f_n \equiv 0$ for $m < n$. Then,
$f_n \stackrel{\text{wk-}\Part}{\to} f$, but $\|f_n - f\|_p = 1$
for all $n$ and all $1 \le p < \infty$. On the other hand,
since $|\int_\Omega f\, | \le \|f\|_p |\Omega|^{1/p^\prime}$, 
$L^p$ convergence implies wk-$\Part$ convergence.
Thus, on $\DensN$, weak-$\Part$ is weaker than $L^p$ ($1\le p < \infty$).

The $f_n$ of the example became more and more oscillatory
as $n$ increased. In our application, that sort of behavior
implies ever increasing kinetic energy.

\subsection{Approximation by projection}
\label{sec:approx-by-projection}

\subsubsection{Weak-$\Part^n$ is weak-but-not-leaky}

Recall that $C_n$ is the union of $\Part^n$-cells contained
fully within the ball of center $0$ and radius $n$, and
$1_{C_n}$ is the indicator function of $C_n$, equal to 1 on $C_n$
and to zero on its complement $C_n^c$.
\begin{lem}
\label{lem:cutoff-proj-approx}
Let $g \in L^1(\Real^3)$. Then, 
both $1_{C_n}\,\Pi_n g$ and $\Pi_n g$ tend to $g$ in $L^1$ as $n\to\infty$.
\end{lem}
\begin{proof}
Consider first a continuous function with bounded support,
$f \in C_c({\mathrm R}^3)$.
It is easy to see that $\|f - \Pi_n f \|_1 \to 0$. For, being compactly 
supported, $f$ is uniformly continuous, so given $\epsilon > 0$, there
is $\delta > 0$ such that $|f(x+y) - f(x)| < \epsilon/M$ 
whenever $|x-y| < \delta$, where $M$ is the Lebesgue measure of the support of $f$.
But then the average of $f$ over a region $A$
with diameter less than $\delta$ differs by less than $\epsilon/M$
from its value at any point in $A$. 
Therefore, if $n$ is large enough that $\sqrt{3} 2^{-n} < \delta$,
it follows that $\|f - \Pi_n f\|_1 < \epsilon$.

Now, given $g \in L^1(\Real^3)$, and $\epsilon > 0$, 
There is $f\in C_c({\Real}^3)$ with $\|g - f\|_1 < \epsilon/3$,
because continuous functions with bounded support are dense in $L^1({\Real})$.
Then, by the triangle inequality, 
\begin{align}
\|g - 1_{C_n}\, \Pi_n \rho\|_1 
&\le 
\|\rho - f\|_1 
+ \|f - \Pi_n f \|_1
+ \|\Pi_n(f - \rho)\|_1
+ \|\Pi_n1_{C_n^c}\, g\|_1
\nonumber \\
&\le 
2\|\rho - f\|_1 
+ \|f - \Pi_n f \|_1
+ \|1_{C_n^c}\, g\|_1.
\nonumber
\end{align}
The final line follows because $\Pi_n$ is an $L^1$ contraction.
In the final line, the second term tends to zero by the previous 
paragraph and the last term simply because $g$ is integrable.
Therefore, $\limsup_{n\to\infty} \|g - 1_{C_n}\,\Pi_n g\|_1 < \epsilon$.
The proof for $\Pi_n g$ is even easier, as the final term in the
displayed inequality is absent in that case.
\end{proof}
This lemma immediately implies that the weak-$\Part$ topology is
weak-but-not-leaky, as shown in the following corrolary. 
However, it may be skipped without loss since a stronger result
is proven independently in Thm. \ref{thm:L1-approx-dichotomy}.
\begin{cor}
The weak-$\Part$ topology on $\DensN$ is weak-but-not-leaky.
\end{cor}
\begin{proof}
That the weak-$\Part$ topology is weaker than $L^1$ is trivial.
Lemma \ref{lem:cutoff-proj-approx} shows that it is Hausdorff.
There only remains to show that if $\rho \in \DensN$ and 
$\rho^\prime$ is $d_\Part$-near $\rho$, that $\rho^\prime$ is tight.
But, $\int_{C_n}|\Pi_n(\rho^\prime - \rho)| \near 0$ implies
that $\int_{C_n} \rho^\prime \near \int_{C_n}\rho$. Since both
integrate to $\Numparts$,
$\int_{C_n^c} \rho^\prime \near \int_{C_n^c}\rho$ for all limited $n$. 
But $\int_{C_n^c}\rho^\prime$ is decreasing in $n$ and $\int_{C_n}\rho \to 0$.
\end{proof}

\subsection{Weak-$\Part$ with internal energy control is as good as $L^1$}
The next theorem shows that if a sequence $\rho_n$ converges with respect to 
$d_{\Part}$, then the only way it can avoid converging with respect to $L^1$ 
is to have divergent internal energy.

\subsubsection{Poincar\'e inequality\cite{Adams,Attouch,Acosta+Duran03,Bebendorf03}}

Let $\Omega$ be a convex bounded region with diameter $\mathrm{diam}(\Omega)$.
Then, if $f$ is in $W^{1,1}(\Omega)$ (it and its 
distributional gradient $\nabla f$ are integrable over $\Omega$),
and denoting the mean of $f$ over $\Omega$ by $\langle{f}\rangle_\Omega$,
the Poincar\'e inequality we need is
\begin{equation}
\int_\Omega |f - \langle{f}\rangle_\Omega| \, d{x} \le \frac{\pi}{2} \, \mathrm{diam }(\Omega)
\int_\Omega |\nabla f| \, d{x}.
\label{poincare}
\end{equation}
Applying this inequality to each cell of $\Part^n$ and summing the results yields
\begin{equation} 
\|\rho - \Pi_n \rho\|_1
\le \frac{\pi\sqrt{3}}{2^{n+1}} \|\nabla \rho\|_1 \le \frac{c}{2^n} F(\rho)^{1/2},
\quad \text{for }\rho \in \DensN.
\label{ineq:coarse-graining-close}
\end{equation} 
Always, of course, $\|\rho - \Pi_n\rho\|_1 \le 2\Numparts$.

\subsubsection{A dichotomy}
\label{sec:dichotomy}

In preparation for Thm. \ref{thm:L1-approx-dichotomy}, we derive a useful inequality.
For any measurable set $B$ and $\rho,\rho^\prime\in \DensN$,
\begin{equation}
\|\rho - \rho^\prime\|_1 
 = \|(\rho - \rho^\prime)1_B\|_1
+ \|(\rho - \rho^\prime)1_{B^c}\|_1,
\nonumber
\end{equation}
Now we estimate the second term on the right-hand side.
\begin{align}
\|(\rho - \rho^\prime)1_{B^c}\|_1
&\le \int_{B^c}\rho\, dx + \int_{B^c}\rho^\prime\, dx
\le
2 \int_{B^c}\rho\, dx + 
\left| \int_{B^c}(\rho^\prime -\rho)\, dx \right|
\nonumber \\
&= 2 \int_{B^c}\rho\, dx + 
\left| \int_{B}(\rho^\prime -\rho)\, dx \right|
\le 2 \int_{B^c}\rho\, dx + 
\|(\rho - \rho^\prime)1_{B}\|_1.
\nonumber
\end{align}
The first inequality is a consequence of the fact that densities
are non-negative. The equality in the second line follows because
both densities $\rho$ and $\rho^\prime$ have the same total integral ($\Numparts$),
and the final inequality is just pulling the absolute value inside the
integral. Plugging this last result into the previous display,
\begin{equation}
\rho,\rho^\prime \in \DensN \Longrightarrow
\|\rho - \rho^\prime\|_1 
\le 2 \|(\rho - \rho^\prime)1_B\|_1
+ 2 \int_{B^c}\rho\, dx.
\label{L1-B-split}
\end{equation}

\begin{thm}
\label{thm:L1-approx-dichotomy}
Suppose $\rho_n \stackrel{wk-\Part}{\to} \rho$ in $\DensN$.
Then, either $\rho_n \stackrel{L^1}{\to} \rho$ 
\hfill\break
or 
{$\liminf_n \Fpure(\rho_n) = +\infty$.}
\end{thm}
\begin{proof}
The equivalent nonstandard phrasing of the lemma is:
if $\rho^\prime$ is $d_\Part$-near $\rho$, then either
it is $L^1$-near $\rho$, or it has illimited internal energy.
There is nothing to show unless $\Star{\Fpure}(\rho^\prime) \ll \infty$, so
assume that is so.

The inequality (\ref{L1-B-split}) plays a key role,
and a triangle inequality,
\begin{equation}
\|\rho - \rho^\prime\|_1 
\le \|\rho - \Pi_n\rho\|_1
+ \|\Pi_n(\rho - \rho^\prime)\|_1
+ \|\rho^\prime - \Pi_n\rho^\prime\|_1,
\label{ineq:triangle}
\end{equation}
will be used much as in Lemma \ref{lem:cutoff-proj-approx}.
We proceed by showing that all three terms on the right-hand side
are infinitesimal for some $N\in\Star{\Nat}$.
By definition of the weak-$\Part$ topology,
$\int_{C_n}|\Pi_n(\rho-\rho^\prime)|\, dx \,\, \near 0$ for every
limited $n$ [recall the definition of the $C_n$ from (\ref{d-Part})].
Robinson's lemma (\S \ref{sec:nsa-overspill}) 
implies existence of an illimited $N$, fixed in 
the following discussion, such that that continues to hold for all $n\le N$.

\noindent 1st term of (\ref{ineq:triangle}):
$\rho$ is standard and $N$ illimited. By
Lemma \ref{lem:cutoff-proj-approx},
$\|\rho - \Pi_N\rho\|_1 \near 0$.

\noindent 2nd term:
Apply (\ref{L1-B-split}) with $C_N$ for $B$, and note
that $\int_{C_N^c}\rho\, \near 0$ ($\rho$ is standard).
Thus, $\|\Pi_N(\rho - \rho^\prime)\|_1 \near 0$.

\noindent 3rd term:
$\rho^\prime$ has limited internal energy.
Therefore, (\ref{ineq:coarse-graining-close}) implies
$\|\rho^\prime - \Pi_N\rho^\prime \|_1 \near 0$. 
\end{proof}

Internal energy functionals are not upper semicontinuous because
by putting in short-wavelength wiggles which are a negligible change
in $L^p$ (for any $p$) norm can drive the internal energy 
arbitrarily high. 
The inequality (\ref{thm:L1-approx-dichotomy}) and
Thm. \ref{ineq:coarse-graining-close} show the flip-side of this
phenomenon: under a bound on internal energy, we really only need
to know a density down to a certain spatial resolution. Below that
scale, it must be almost flat.

\section{Ground-state energy, or `energy floor' ($E$)}
\label{sec:energy-floor}
\setcounter{subsubsection}{0}

At long last, we are ready to begin considering external
one-body potentials.
A one-body potential $v$ induces the $\Numparts$-body external potential
$\underline{v}(\underline{x}) = \sum_{i=1}^\Numparts v(x_i)$.
The usual way of defining the ground-state energy is
\begin{equation}
E(v) = \inf
\left\{ \EE(\psi) + \langle \psi | \underline{v} \psi \rangle 
\,:\, \psi\in\PureStates \right\}.
\nonumber
\end{equation}
Of course, the term $\langle \psi | \underline{v} \psi \rangle$ only depends
on the state through its density, so the minimization can done in
two stages:
\begin{equation}
E(v) = \inf
\left\{ \Fpure(\rho) + \pair{v}{\rho} \,:\, {\rho\in\DensN} \right\}.
\label{E-variational}
\end{equation}
Here, we have introduced the notational abbreviation
\begin{equation}
\pair{v}{\rho} \defeq \int v\, \rho.
\end{equation}
This is common in situations of pairing between elements of a Banach space
and its dual, but we will use it even in cases which may yield $\pair{v}{\rho} = +\infty$.
Now, if the external potential $v$ is bounded, then $\int v \rho$ is well-defined 
for all $\rho \in \DensN$ and the minimization in (\ref{E-variational}) 
makes sense. But if it is unbounded, we must proceed more cautiously.

\subsubsection{on the term ``ground-state energy''}
\label{sec:term-ground-state-energy}

The term ``ground-state energy'' for $E(v)$ is customary in the
DFT literature. But, unless the infimum is actually a minimum, 
there is no genuine {\em ground state}. $v \equiv 0$ is the simplest
example of a potential for which there is no minimum.
A term, such as `energy floor', which does not imply the existence of 
a ground state might therefore be more appropriate.
Mostly, we shall avoid the issue by simply writing `$E$'.

\subsubsection{Habitability and stability}

\begin{defn}
An external potential $v$ is {\it habitable} if $E(v) < +\infty$ and
{\it stable} if $E(v) > -\infty$.
\end{defn}
Needless to say, we are interested in external potentials that are
both habitable and stable, but, ensuring these two conditions involves
very differing considerations.

\subsubsection{some notation}

$\pair{v}{\rho}$ is not a function; it is the value of the
function $x \mapsto \pair{v}{x}$ at the argument $\rho$.
To deal with this awkwardness, we write $\pair{v}{\cdot}$ for the 
function $\rho \mapsto \pair{v}{\rho}$.
The notation
\begin{equation}
v^+ \defeq v \vee 0, \,\,\, v^- \defeq (-v) \vee 0
\nonumber
\end{equation}
for the positive and negative part of an external potential,
such that $v = v^+ - v^-$ and
$\pair{v}{\cdot} = \pair{v^+}{\cdot} - \pair{v^-}{\cdot}$,
will also be useful in this section.

\subsection{Bounded external potentials}
\label{sec:E-bdd-potl}

We begin by restricting our attention to external potentials in
$L^\infty$, for which we use the notation `$\PotlsBdd$'.

If $w\in \PotlsBdd$, then $\pair{w}{\cdot}$ is
$L^1$ {\em continuous}, since $L^\infty$ is the Banach
space dual of $L^1$. That is, 
$|\pair{w}{\rho} - \pair{w}{\rho^\prime}| \le \|w\|_\infty \|\rho-\rho^\prime\|_1$.
And since $F$ is $L^1$ lsc, so too is $F+\pair{w}{\cdot}$. 
However, Cor. \ref{cor:F-is-lsc} established that $F$ is lsc with respect
to much weaker topologies, so we should try to show the same for 
$F+\pair{w}{\cdot}$. The full class of weak-but-not-leaky topologies
is beyond our reach, but the weak-$\Part$ topology can be handled.
\begin{prop}
\label{prop:F+PotlsBdd-lsc}
For $w \in \PotlsBdd$, 
$F + \pair{w}{\cdot}: 
(\DensN,\text{wk-}\Part) \to \overline{\Real}$ is lower semicontinuous.
\end{prop}
\begin{proof}
Take $\rho$ in $\DensN$ and $\rho^\prime$ wk-$\Part$-near $\rho$.
According to Lemma \ref{thm:L1-approx-dichotomy}, either $F(\rho^\prime) \near \infty$,
or $\rho^\prime$ is $L^1$-near $\rho$. In the former case, of course
$F(\rho^\prime) \gtrsim F(\rho)$,
because $|\int w \rho^\prime | \le \Numparts \|w\|_\infty$ is bounded.
In the latter, the conclusion follows by $L^1$ lsc of $F$ and $L^1$
continuity of $\pair{w}{\cdot}$.
\end{proof}

\subsubsection{Fenchel conjugate and 
biconjugate\cite{ET,Aubin+Ekeland,Borwein+Zhu,Phelps88}}
\label{sec:biconjugate}

We take a break to do some general geometric analysis.
A proper (not identically $+\infty$) function 
$f: X \rightarrow \overline{\Real}$ on a vector space $X$
is given, along with a vector space $Y$ of linear functionals $X\rightarrow\Real$.
The topology on $X$ is that induced by $Y$, denoted $\sigma(X,Y)$. 
By definition, then, a set is a neighborhood of the origin if and only
if it contains an intersection 
$\cap_n \{x : |\pair{\xi_n}{x}| < \epsilon_n\}$, 
with $\xi_n \in Y$ and $\epsilon_n > 0$. In the $\sigma(X,Y)$ topology, 
a point $x_0$ is outside a closed convex set
$A$ if and only if there is $(\zeta,c) \in Y\times\Real$ such that
$\pair{\zeta}{x} > c > \pair{\zeta}{x_0}$ for all $x \in A$, that is, 
$x_0$ can be strictly separated from $A$ by a linear functional in $Y$.
Here is an explanation for that.
Forgetting about $A$'s convexity for the moment, there
must be a collection $\{(\xi_i,c_i): 1\le i\le n\}$ and $\epsilon > 0$
such that $c_i = \xi(x_0)$, and for every $x \in A$, there is
some $i$ such that $|\pair{\xi_i}{x} - c_i| > \epsilon$.
Arrange the $\xi_i$'s into a map 
$\underline{\xi} = (\xi_1,\ldots,\xi_n) : X \rightarrow \Real^n$
into a finite-dimensional Euclidean space. Then, the image of
$A$, $\underline{\xi}(A)$ is outside the $\epsilon$ ball around
$\underline{c} = (c_1,\ldots,c_n)$. But if $A$ is convex, then
so is its image in $\Real^n$. We can apply {\em finite-dimensional} 
separation properties in $\Real^n$ to obtain a linear combination
$\sum a_i \xi_i$ which achieves the desired separation.

The function defined on $Y$ by
\begin{equation}
f^*(\zeta) \defeq \sup
\left\{ \pair{\zeta}{x} - f(x): x\in X \right\},
\end{equation}
is called, among other names, the {\it Fenchel conjugate}.
Up to signs, the relationship between
$f$ and $f^*$ is essentially the same as that between $F$ and $E$.
Note that 
\begin{equation}
\forall x\in X.\,\,
f(x) \ge c + \pair{\zeta}{x}
\,\,\, \Leftrightarrow \,\,\,
-f^*(\zeta) \ge c. 
\nonumber
\end{equation}
Defining the half-space indexed by $(\zeta,c) \in Y\times\Real$ by
\begin{equation}
H(\zeta,c) \defeq \left\{ (x,z) \in X\times\Real : z \ge c + \pair{\zeta}{x} \right\},
\end{equation}
we see that $\epi f \subseteq H(\zeta,c)$ if and only if 
$-f^*(\zeta) \ge c$.
(Recall that the epigraph of $f$ is the set
$\epi f = \left\{ (x,z) \in X\times\Real : f(x) \le z \right\}$
of points on or above the graph of $f$.)

The proposition we are aiming for is this:
with $\hat{f}$ the function satisfying $\epi \hat{f} = \cco \epi f$,
\begin{equation}
\hat{f}(x) = 
\sup
\left\{ \pair{\zeta}{x} - f^*(\zeta): {\zeta\in Y} \right\} 
= f^{**}(x).
\label{eq:biconjugate}
\end{equation}
An equivalent statement is that
$\cco \epi f$ is equal to the intersection of all $H(\zeta,c)$ that contain
$\epi f$, and it is in this form that we prove it, by showing that
any point {\em not} in $\epi f$ falls outside some such half-space.
We prove this assuming that $f \ge 0$ (adequate for our purposes), 
and later comment on the general case.
Suppose $(x_0,z_0) \not\in \cco \epi f$. Then, there is 
$(\zeta,\beta)\in Y\times\Real$, $c\in\Real$ and $\epsilon > 0$ such that
\begin{equation}
\forall (x,z) \in \epi f.\,\,
\pair{\zeta}{x} + \beta z > c +\epsilon > c > \pair{\zeta}{x_0} + \beta z_0.
\nonumber
\end{equation}
Since we can always increase $z$ without leaving $\epi f$, $\beta \ge 0$.
If $x_0 \in \dom f$, then $(x_0,z) \in \epi f$ for large enough $z$,
so $\beta >0$ in that case.
If $\beta > 0$, dividing through by $\beta$, we see that
$\epi f$ is in $H(-\zeta/\beta,c/\beta)$, but $(x_0,z_0)$ is not.
The other case is $\beta = 0$. Then,
\begin{equation}
c + \pair{-\zeta}{x} <  0 <
c + \pair{-\zeta}{x_0}, \,\, 
\forall (x,z) \in \epi f.
\nonumber
\end{equation}
Thus, since $f \ge 0$, $f(x) > Mc + \pair{-M\zeta}{x}$ for any $M > 0$.
But, for large enough $M$, we will certainly have $z_0 < Mc + \pair{-M\zeta}{x_0}$.
That completes the proof assuming $f \ge 0$.
If that is not the case, since $f$ is proper,
there is some $x_0 \in \dom f$. The $\beta$ we get for this $x_0$ cannot
be zero, as remarked earlier. Hence,
$f + \pair{\zeta/\beta}{\cdot} - c/\beta \ge 0$. 
If the proposition is true for this tilted and shifted function, 
it is also true for $f$.

\subsubsection{application to $F$}
\label{sec:F-sufficiency}

Recall that $\epi F = \cco \epi \Fpure$ (Thm. \ref{thm:F-is-cco-of-Fpure}),
where the closure is with respect to any weak-but-not-leaky topology on
$\DensN$, such as $L^1$ or weak-$\Part$.
Now, the weak-$\Part$ topology is generated by the indicator functions of
cells of $\Part$. Thus, to apply \ref{sec:biconjugate}
we take $Y$ to be
\begin{equation}
\PotlsSimple \defeq \mathrm{span} \left\{1_\Omega \,:\, \Omega \in \Part \right\},
\end{equation}
the vector space of functions spanned by the indicator functions of
cells of $\Part$. A potential in $\PotlsSimple$ is nonzero on only a
finite number of cells. Then,
\begin{equation}
F(\rho) = \sup
\left\{ E(v) - \pair{v}{\rho} : {v \in \PotlsSimple} \right\}.
\label{eq:simple-potls-are-sufficient}
\end{equation}
If we call a set of external potentials {\it $F$-sufficient} 
just in case the supremum over
the functions $E(v) - \pair{v}{\cdot}$ is $F$, 
then the display says that $\PotlsSimple$ is $F$ sufficient. 
Since always
\begin{equation}
E(v) \le F(\rho) + \pair{v}{\rho},
\end{equation}
any larger class of external potentials is also $F$-sufficient. 
We are not claiming that $\PotlsSimple$ is a minimal $F$-sufficient class,
but it is certainly pretty small.
\begin{rem}
Starting from $L^1$ lower-semicontinuity of $F$, one may similarly show the
more-traditional, but weaker, result that $L^\infty$ is $F$-sufficient. The proof 
uses the Hahn-Banach principle in the form of the assertion that the 
closures of a convex set relative to the initial and weak topology are the same.
\end{rem}
\subsection{$[0,+\infty]$-valued external potentials}
\label{Potl-up}

We now move from the relatively modest space $\PotlsBdd$ of
bounded external one-body potentials to a much larger class $\PotlsUp$.
In contrast to the perturbations we consider later, it
is not a vector space, but only a convex cone.
This class consists of all measurable external potentials with values in
the non-negative extended reals:
\begin{equation}
\PotlsUp = (\Real^3 \rightarrow [0,+\infty]).
\end{equation}
Why deal with such a large class of potentials?
First, they have practical use in modeling
various sorts of confining potentials, such as a harmonic oscillator 
potential or even an infinite cubic-well.
Neither of these is in any $L^p$ space, and
a nice feature of our treatment is that such confining potentials
are integrated into the {\em density} functional picture without
any need to go back to talking explicitly about wavefunctions.
Second, if we are going to deal with external potentials unbounded
above, there is hardly any extra work to go all the way, for what
we do. Third, and closely related to the second point, the differing
{\em physical} significances of unboundedness above and below for
an external potential is brought into the starkest possible relief.

A potential $w\in\PotlsUp$ can be approximated density-wise
by {\em bounded} potentials.
That is, with the truncations 
\begin{equation}
(w\wedge n)(x) \defeq (\text{if } w(x) < n \text{ then } w(x) \text{ else } n),
\label{eq:truncation-def}
\end{equation}
we have
\begin{equation}
\langle w\wedge n,\rho\rangle \nearrow
\langle w,\rho \rangle, \,\, \text{as } {n\to\infty}.
\nonumber
\end{equation}
Thus,
\begin{equation}
F(\rho) + \langle w\wedge n,\rho\rangle \nearrow
F(\rho) + \langle w,\rho \rangle.
\nonumber
\end{equation}
The point is that if we can show that 
$\rho \mapsto F(\rho) + \langle w \wedge n,\rho\rangle$ is lower semicontinuous
for each $n$, then so is the supremum (\S \ref{sec:lsc-def}).
But, we've already done that.
\begin{thm}
\label{thm:F+PotlsUp-lsc}
For $w \in \PotlsUp$, 
$F_w \defeq F + \pair{w}{\cdot}: 
(\DensN,\text{wk-}\Part) \to \overline{\Real}$ is convex and lower semicontinuous.
\end{thm}
\begin{proof}
By Prop. \ref{prop:F+PotlsBdd-lsc} and the preceding discussion.
\end{proof}

\subsubsection{sufficient condition for habitability of $\PotlsUp$ external 
potentials}
\label{sec:habitability-condition}

\begin{prop}
\label{prop:habitability-condition}
$w \in \PotlsUp$ is habitable
if $\int_U w < \infty$
for some open set $U\in\Real^3$.
\end{prop}
\begin{proof}
Since $U$ contains an open ball, one can easily find a smooth 
$\rho$ with support in $U$, bounded by $\|\rho\|_\infty < \infty$. 
By inequality (\ref{ineq:F-equiv-H1-sqrt-rho}), $F(\rho) < \infty$,
and $\pair{w}{\rho} \le \|\rho\|_\infty \int_U w < \infty$.
\end{proof}

\subsection{$F$-norm, and relatively $F$-small perturbations}
\label{sec:F-norm}

An potential unbounded below poses a danger of {\em instability}.
For such a potential, we can find a sequence $\rho_n$ of densities
such that $\pair{v}{\rho_n} \to -\infty$. Intuitively, this is achieved
by squeezing the density more and more tightly into a bottomless potential
well. To assure that the {\em total} energy cannot be driven
to $-\infty$ in this way, we make sure that the kinetic energy cost of
such a squeezing eventually dominates. 
This motivates the following definition.
\begin{defn}
The {\it $F$-norm} of a potential $v$ is
\begin{equation}
\|v\|_F \defeq \inf\left\{ a\in (0,\infty] : 
\forall \rho \in \dom F.\,
\pair{|v|}{\rho} \le a F(\rho)
\right\}.
\end{equation}
A potential $v$ is {\it $F$-small} if $\|v\|_F < 1$,
and {\it relatively $F$-small} if it has a decomposition $v = v_1 + v_2$, with
$v_1$ $F$-small and $\|v_2\|_\infty < \infty$.
The set of all $F$-small potentials is denoted 
$\PotlsF$, and the subset of non-negative ones, ${\PotlsF}^{+}$.
\end{defn} 
Thus, the set of relatively $F$-small potentials is $\PotlsF + \PotlsBdd$, but
there is no implication that the decomposition is unique.

The appelation `$F$-norm' is justified because the norm conditions
$\|v + w\|_F \le \|v\|_F + \|w\|_F$, $\|a v\|_F = |a| \|v\|_F$ ($a$ real)
and $\|v\|_F = 0 \Rightarrow v \equiv 0$ are satisfied. 
These imply that $\|\cdot\|_F$ is convex. 
\begin{lem}
$\PotlsF$ is convex and balanced
\rm{(}i.e., $x\in \PotlsF$ implies $-x \in \PotlsF$\rm{)}.
\end{lem}

With that, we are ready to define our big space of potentials.
\begin{defn}
\begin{equation}
\Potls 
\defeq  \PotlsUp + \PotlsBdd + \PotlsF.
\label{eq:big-potl-space}
\end{equation}
\end{defn}
\begin{thm}
\label{thm:stability}
Every external potential in $\Potls$ is stable.
\end{thm}
\begin{proof}
Noting that $\Potls$ can also be split as $\PotlsUp + \Const - \PotlsF^+$,
let $v$ be decomposed as $v = w + u - u^\prime \in \PotlsUp + \Const - \PotlsF^+$.
Then, since $w+u$ is bounded below, so is $\pair{w+u}{\rho}$ independently of 
$\rho \in \DensN$.
Also,
\begin{equation}
F(\rho) - \pair{u^\prime}{\rho} \ge (1 - \|u^\prime\|_F) F(\rho) > 0.
\nonumber 
\end{equation}
\end{proof}

\subsection{attitude and topology}
\label{sec:potential-topology}
Our attitude is not to view $\Potls$ as monolithic. 
Rather, the idea is that there is a habitable background potential $w \in \PotlsUp$, 
essentially regarded as fixed, and a variable additional part in $\PotlsBdd + \PotlsF$.
This idea was already signalled by the notation
$F_w \defeq F + \pair{w}{\cdot}$ introduced in Thm. \ref{thm:F+PotlsUp-lsc}.
If $w$ is habitable, then $F_w$, in addition to being convex and lower semicontinuous,
is also proper (somewhere less than $+\infty$, everywhere greater than $-\infty$).
$F_w$ is thus well-suited to be considered a modification of $F$ which can take the
place of the latter in the theory. Accordingly, 
for $w\in \PotlsUp$ habitable, we write
\begin{equation}
E_w(v) = \inf \left\{ F_w(\rho) + \pair{v}{\rho} \,:\, {\rho\in\DensN} \right\},\quad
v \in \PotlsBdd + \PotlsF.
\label{eq:E_w-minimization}
\end{equation}
A natural question is whether $E_w(v)$ is continuous as a function of $v$, but to make
it meaningful, $\PotlsBdd + \PotlsF$ must have a topology.
We will use the norm
\begin{equation}
\|v\| = \inf \left\{ \|v_1\|_\infty + \|v_2\|_F\,;\, v_1 + v_2 = v \right\}.
\end{equation}
However, insofar as it is important to have the second member, $v_2$, of
such a split be in $\PotlsF$, this norm will be mostly used just to define small 
neighborhoods of potentials, rather than to measure the ``size'' of individual potentials. 

Returning to the minimization problem (\ref{eq:E_w-minimization}), 
note that with $v = v_1 + v_2 \in \PotlsBdd + \PotlsF$,
\begin{equation}
\pair{v}{\rho} \ge -\Numparts  \|v_1\|_\infty - \|v_2\|_F F(\rho),
\nonumber
\end{equation}
for any density. Thus,
\begin{equation}
-\Numparts  \|v_1\|_\infty + [ 1 - \|v_2\|_F ] F_w(\rho) 
\le
F_w(\rho) + \pair{v}{\rho} 
\le 
\Numparts  \|v_1\|_\infty + [ 1 + \|v_2\|_F ] F_w(\rho). 
\nonumber
\end{equation}
Now, let $\rho_0$ be an inhabitant of $w$ (i.e., $F_w(\rho_0) < \infty$). 
In the minimization, there is no point in considering any density with
$F_w(\rho) + \pair{v}{\rho} > F_w(\rho_0) + \pair{v}{\rho_0}$. 
Therefore, attention may be confined to densities $\rho$ satisfying
the search region condition
\begin{equation}
F_w(\rho) \le
\frac{1}{ 1 - \|v_2\|_F }\Big(
2\Numparts  \|v_1\|_\infty + [ 1 + \|v_2\|_F ] F_w(\rho_0) \Big).
\label{eq:F-search-radius}
\end{equation}
This condition will be very important when we come to discuss continuity,
in Section \ref{sec:E-continuity}.

\subsection{Lieb potentials: $L^{3/2} + L^\infty$}
\label{Lieb-potls}

Conceptually, the $F$-norm suits our purposes perfectly. Unfortunately, its
computation poses significant, if not insurmountable, difficulties. 
The space of potentials $L^{\infty} + L^{3/2}$ equipped with the norm
$\|v\| = \inf \left\{ \|v_1\|_\infty + \|v_2\|_{3/2}\,;\, v_1 + v_2 = v \right\}$,
does not have that problem. 
Since this space was introduced by Elliott Lieb, we will denote it by `$\Lieb$'.

Recall the inequality (\ref{L3-F-bnd}): $\|\rho\|_3 \le c F(\rho)$.
Since 
$\pair{|v|}{\rho} \le \|v\|_{3/2}\|\rho\|_3 \le c \|v\|_{3/2} F(\rho)$,
\begin{equation}
\label{eq:F-vs-3/2-norm}
\|v\|_F \le c\|v\|_{3/2}.  
\end{equation}
This does not say that an $L^{3/2}$ potential is $F$-small. But,
as the next Lemma shows, if we chop off a big enough bounded part, 
the remainder will be so.
\begin{prop}
\label{prop:L3/2-F-small}
For any $\epsilon > 0$, there is a continuous embedding
$\Lieb \hookrightarrow \PotlsBdd + \epsilon\cdot \PotlsF$.
\end{prop}
\begin{proof}
Continuity is a simple consequence of (\ref{eq:F-vs-3/2-norm})
once the inclusion is shown valid. We demonstrate the inclusion.

$v^+$ and $v^-$ can be treated separately, so, without loss, 
assume $v \ge 0$. For any $n > 0$, 
$v = (v\wedge n) +   v - (v\wedge n)$, $\|v\wedge n\|_\infty \le n$ and
$v\wedge n \to v$ in $L^{3/2}$ norm as $n \to \infty$.
So, given $\epsilon$, there is an $n$ such that 
$c \| v - (v\wedge n) \|_{3/2} < \epsilon$.
\end{proof}

\subsection{Concavity}
\label{sec:E-concave}

A function $f$ is {\it concave} if $-f$ is convex. So $E$ is concave
if its domain is convex and for all $v$ and $w$ in its domain,
\begin{equation}
\forall \alpha \in (0,1).\,\, E(\alpha v + (1-\alpha) w) \ge 
\alpha E(v) + (1-\alpha) E(w).
\nonumber
\end{equation}
The formula (\ref{E-variational}) expresses $E$ as the 
infimum of a collection of functions $F(\rho) + \pair{\cdot}{\rho}$
indexed by densities. 
That can also be written as
\begin{equation}
\sub E = \bigcap_{\rho\in\DensN} \sub \left(F(\rho) + \pair{\cdot}{\rho}\right).
\nonumber
\end{equation}
Here, we use the notion of {\it subgraph}:
$\sub f$, the subgraph of $f$, is the region on or below the
graph of $f$. $f$ is a concave function if $\sub f$ is a convex set. 
If each of the functions $F(\rho) + \pair{\cdot}{\rho}$ is concave,
then all their subgraphs are convex, hence the intersection of the subgraphs 
is too, and therefore $E$ is concave.
To demonstrate concavity of $E$, it suffices to verify that
$F(\rho) + \pair{\cdot}{\rho}$ is concave for each $\rho \in \DensN$.
That would be trivial except for the fact that $\pair{v}{\rho}$ can be
$-\infty$ for some potentials.
Given $v,w \in \Potls$ and $0 < \alpha < 1$, we need to show
\begin{equation}
F(\rho) + \pair{\alpha v + (1-\alpha) w}{\rho} \ge
\alpha \left[ F(\rho) + \pair{v}{\rho}\right]
  + (1-\alpha)\left[ F(\rho) + \pair{w}{\rho} \right].
\nonumber
\end{equation}
Decompose $v = v^\prime + v^{\prime\prime}$ with $v^\prime \in \PotlsUp + \Const$
and $v^{\prime\prime} \in \PotlsF$, and similarly, 
$w = w^\prime + w^{\prime\prime}$.
Since 
\begin{equation}
\pair{\alpha v^\prime + (1-\alpha) w^\prime}{\rho} =
\alpha \pair{v^\prime}{\rho}  + (1-\alpha) \pair{w^\prime}{\rho},
\nonumber
\end{equation}
the problem reduces to showing
\begin{equation}
F(\rho) - \pair{\alpha v^{\prime\prime} + (1-\alpha) w^{\prime\prime}}{\rho} \ge
\alpha \left[ F(\rho) - \pair{v^{\prime\prime}}{\rho}\right]
  + (1-\alpha)\left[ F(\rho) - \pair{w^{\prime\prime}}{\rho} \right].
\nonumber
\end{equation}
But, since $v^{\prime\prime}$ and $w^{\prime\prime}$ are $F$-small,
if either potential term on the right-hand side is $-\infty$
then $F(\rho) = +\infty$, so both sides are infinite. 
Thus, the following is now established.
\begin{prop}
$E : \Potls \rightarrow \overline{\Real}$ is concave.
\end{prop}

\subsection{Continuity}
\label{sec:E-continuity}

In keeping with the idea of regarding 
$w+v$ for $w \in \PotlsUp$ and $v\in \PotlsBdd + \PotlsF$
as a perturbation of the background $w$, we write
\begin{equation}
E_w(v) = E(w+v).
\label{eq:bckgd-notation}
\end{equation}
The aim of this section is to show that $E_w(v)$ is continuous as a function of $v$.
What this means, in the standard idiom, is the following. Given $v_0 \in \PotlsBdd + \PotlsF$,
for any $\epsilon > 0$, there is a $\delta > 0$ such that 
whenever $v - v_0$ can be split as
\begin{equation}
v - v_0 = v^\prime + v^{\prime\prime}, \quad
\|v^\prime\|_\infty < \delta, \quad
\|v^{\prime\prime}\|_F < \delta,
\nonumber
\end{equation}
then, \hbox{$|E_w(v) - E_w(v_0) | < \epsilon$.}
\begin{thm}
\label{thm:E-continuity}
If $w \in \PotlsUp$ is habitable, then $E_w : \PotlsBdd + \PotlsF \rightarrow \Real$
is continuous.
\end{thm}
\begin{proof}
For convenience, we use the nonstandard characterization of continuity.
Translation to $\epsilon$-$\delta$ form is straightforward.
The condition (\ref{eq:F-search-radius}) plays a key role. 
Split standard potential $v_0$ as
\begin{equation}
v_0 = v_1 + v_2, \quad
v_1 \in \PotlsBdd, \quad 
v_2 \in \PotlsF,
\nonumber
\end{equation}
Now suppose $v$ (in $\Star{\PotlsBdd} + \Star{\PotlsF}$) is infinitely close to $v_0$, 
so $v - v_0 = v^\prime + v^{\prime\prime}$, with
$\|v^\prime\|_\infty \simeq 0$, and $\|v^{\prime\prime}\|_F \simeq 0$.
Then, 
\begin{equation}
|F(\rho) + \pair{v}{\rho} - [F(\rho) + \pair{v_0}{\rho}]| 
\le \Numparts \|v^\prime\|_\infty +  F(\rho) \|v^{\prime\prime}\|_F,
\label{eq:close}
\end{equation}
and the right-hand side is infinitesimal, as long as $F(\rho)$ is limited.
For that restriction, we appeal to the search region condition (\ref{eq:F-search-radius}).
Noting in particular that $\|v_2 + v^\prime\|_F \ll 1$ (more that infinitesimally less), 
we see that the minimizations required to find $E_w(v_0)$ and $E_w(v^\prime)$ can 
proceed under a common finite bound on $F(\rho)$. Therefore, the left-hand side of 
(\ref{eq:close}) is infinitesimal over the relevant portion of $\DensN$, 
and $E_w(v^\prime) \simeq E_w(v_0)$, which is what was to be proved.
\end{proof}

\begin{cor}
\label{thm:E-Lieb-continuity}
$E_w : \Lieb \rightarrow \Real$ is continuous.
\end{cor}
\begin{proof}
According to Prop. \ref{prop:L3/2-F-small}, the map in question is the composite
\begin{equation}
\Lieb \hookrightarrow \PotlsBdd + \PotlsF \stackrel{E_w}{\rightarrow} \Real
\nonumber  
\end{equation}
of two continous functions.
\end{proof}

This brings us to the end of our investigation of the 
functionals $F$ and $E$. In \ref{sec:F-sufficiency} we just
scratched the surface of the relationship between them,
and have not even touched the vexed $v$-representability 
problem of finding a potential with given ground-state density.
Nevertheless, it is time to quit, for now.

\appendix
\section{Selective functional analysis refresher}
\label{app:top}

This appendix contains a refresher on some basic notions of
functional analysis. The reader for whom this is completely
new is likely to feel the need for additional resources.

\subsection{Lebesgue ($L^p$) and Sobolev ($H^1$, $W^{1,1}$) spaces}

We recall here the definitions of $L^p$ norms. For any measurable
function $f: {\Real}^n \rightarrow \Real$,
\begin{align}
\|f\|_{L^p}^p & = \|f\|_p^p \defeq \int |f(x)|^p \, dx, \quad 0 < p < \infty
\nonumber \\
\|f\|_{L^\infty}  &= \|f\|_\infty \defeq \inf \{a \ge 0 : \mu_{\mathrm{Leb}}(|f| \le a) > 0\}.
\end{align}
In the second line, $\mu_{\mathrm{Leb}}$ denotes Lebesgue measure.
For a continuous function, $\|\cdot\|_\infty$ is just the supremum of
the absolute value. The complicated business with $\mu_{\mathrm{Leb}}$
is needed so that changing a function on a set of measure zero does
not change its norm.
The space $L^p({\Real}^n)$ consists of the functions with finite
$L^p$ norm. Actually, this is a white lie. 
None of the $\|\cdot\|_p$ can distinguish functions which differ on sets 
of measure zero,
so we really should say that the elements are {\em equivalence classes of functions}
rather than functions, {\em tout court}. But nobody does.

Sobolev spaces involve also norms of derivatives. These are really
distributional derivatives, and for our purposes, it is simplest to
think in terms of Fourier transform. With
\begin{equation}
\hat{f}(q) = \lim_{R\to\infty}\int_{|x|\le R} e^{iq\cdot x} f(x) \, 
\frac{dx}{(2\pi)^{n/2}}
\nonumber
\end{equation}
the Fourier transform of $f$, the $H^1$ norm of $f$ is
\begin{equation}
\|f\|_{H^1}^2
\defeq \int (1 + |q|^2) |\hat{f}(q)|^2 \, {dq}.
\nonumber
\end{equation}
If $f$ is actually differentiable, then this is the same as
\begin{equation}
\|f\|_{H^1}^2 \defeq \int (|f|^2 + |\nabla f|^2) \, dx
= \|f\|_{2}^2 + \|\nabla f\|_2^2,
\nonumber
\end{equation}
where the derivative is understood in the classical sense.
Otherwise $\nabla f$ is a distributional derivative.
The other Sobolev space used in these notes is $W^{1,1}$:
\begin{equation}
\|f\|_{W^{1,1}} \defeq \int (|f| + |\nabla f|) \, dx.
\nonumber
\end{equation}
The superscripts in `$W^{1,1}$' indicate maximum derivative order
and the power to which things are raised in the integral.
$H^1$ is also called $W^{1,2}$. The reader will see that these
are just the beginning of a family of function spaces.
Many sources are available for more information, for example
\cite{Taylor1,Treves-PDE,Folland-PDE,Adams,Lieb+Loss}.

\subsection{the concept of `topology'}

`Topology' is a badly overloaded word, largely due to
a tendency to deploy it in contexts (turning coffee cups 
into a doughnuts) where a more precise term would be better.
As we use it here, {\em a topology} is essentially a notion
of convergence. In a metric space, a sequence $(x_n)$ converges
to $x$ if the distances $d(x_n,x) \to 0$ in $\Real$.
What is involved here are certain standards of closeness,
technically called {\em neighborhoods}.
The ball $B_x(a)$ of radius $a$ about $x$ is one neighborhood of $x$
and the ball $B_x(b)$ is another.
If $a > b$, the second standard of closeness is strictly 
more demanding than the first; $B_x(b) \subset B_x(a)$.
We {\em could} consider ``lopsided'' neighborhoods which do not nest.
The point for a metric space is that we do not need such things;
the notion of convergence is completely and cleanly captured by the
centered balls. But, such a simple structure is not always obtainable. 
Consider the space of bounded functions
$\Real \rightarrow \Real$ where we say that $f_n \to f$ if and
only if $f_n(x) \to f(x)$ for each $x$.
Example: $f_n(x) = (\text{ if }\, x > n \,\text{ then }\, 1 \,\text{ else }\, 0)$
converges to the function $0$, which is identically equal to zero.
$\{g\,:\, |g(a)| < \epsilon\}$ is one neighborhood of $0$ and
$\{g\,:\, |g(b)| < \epsilon\}$ is another. 
The two are not comparable, and we cannot rejigger to eliminate
this complication.
But, there is a third neighborhood, 
$\{g\,:\, |g(x)| < \epsilon \,\text{ for }\, x\in\{a,b\}\}$
which is contained in both. 
This illustrates the general pattern. One way to specify the
topology of a topological space is to give, for each point $x$,
a {\em neighborhood base} ${\mathcal B}_x$ such that if
$A,B \in {\mathcal B}_x$, then there is $C \in {\mathcal B}_x$ 
satisfying $C \subseteq A$, $C\subseteq B$.
Often, given such a system of neighborhood bases, we can
find a metric which has the same concept of convergence.
This is done in \S \ref{sec:metric-for-wk-P}, for example.
In fact, no non-metrizable topology is essential anywhere in
these notes, so the reader loses nothing by reading every
occurence of ``topology'' in the text as 
``notion of convergence associated to some metric''.

\section{Infinitesimals}
\label{sec:nsa}

\subsection{seriously}
\label{sec:nsa-seriously}

Infinitesimals have a long history in mathematics. 
Leibniz based his version of the calculus on them,
and Euler made masterful use of infinitesimals and
their reciprocals, infinitely large quantities.
But the logical status of old-style infinitesimals was ever
shaky, and they were banished from mathematical 
discourse during the 19th century, replaced by the 
$\epsilon$-$\delta$ method.
Eventually, infinitesimals were revived in a rigorous form,
primarily through the work of Abraham Robinson in the 1960's,
giving us modern {\it infinitesimal analysis}, or
{\it nonstandard analysis}\cite{Robinson66}.
Although there are alternate axiomatic treatments\cite{Nelson77,Nelson87,Kanovei+Reeken},
we follow the model theoretic 
approach pioneered by Robinson\cite{Robinson66,Lindstrom85,Goldblatt,Loeb+Hurd,Loeb,Davis,AFHL,Stroyan+Luxemburg}.
The idea is to extend a conventional mathematical universe 
${\Univ}$ to a nonstandard universe $\Star{{\Univ}}$ 
which has all the same {\em formal} properties as ${\Univ}$ does, 
and yet is chock full of ideal elements. That an element is ideal, for
example infinitesimal, is an {\em external} judgement not available
to the {\em internal} object language.
The interplay of internal and external is responsible
for the power of infinitesimal methods. We give here an informal
exposition. The previously-cited references give details
of the customary construction of nonstandard universes by means of 
ultraproducts over set-theoretic superstructures.
Relieved of the burden of proving anything, we describe matters
instead in a type-theoretic way inspired by modern functional
programming languages such as ML and Haskell.
A good strategy for the reader might be to go through 
just \ref{sec:nsa-base}--\ref{sec:nsa-metric}, and then return
when necessary. The only nonstandard argument which is not
elementary is contained in the proof of Lemma \ref{lem:wild-mixed-state}.

\subsection{base types}
\label{sec:nsa-base}

We begin with the base types of Booleans
$\Bool = \{\mathrm{T},\mathrm{F}\}$,
natural numbers $\Nat$ and real numbers $\Real$.
These types, and their elements, are denizens of the mathematical universe
${\Univ}$. The nonstandard universe $\Star{\Univ}$ has 
$\Star{\Bool}$, which is actually identical to $\Bool$, the
hypernaturals $\Star{\Nat}$ extending $\Nat$ and the hyperreals
$\Star{\Real}$ extending the reals. $\Star{\Nat}$ and $\Star{\Real}$ have
all the order and arithemetic properties of $\Nat$ and $\Real$.
But these sets have additional elements; the new elements in $\Star{\Nat}$ 
are larger than all naturals and some of the hyperreals are infinitesimal.
With the understanding of implicit coercion from $\Nat$ to $\Real$
and from $\Star{\Nat}$ to $\Star{\Real}$, we have the following
grab-bag of terminology and ideas.
The essential image is of a number line with a ``halo'' of nonstandard 
numbers around every standard real and then extended on the far left
and far right with ``infinitely large'' numbers.
\begin{defn}
\label{def:nonstd}
Let $x,y \in \Star{\Real}$.

\renewcommand{\labelenumi}{\theenumi}
\renewcommand{\theenumi}{\alph{enumi}.}
\begin{enumerate}
\item $x$ is {\em standard} if $x \in \Real$. Otherwise,
$x$ is {\em nonstandard}.

\item $x$ is {\em limited} if $|x| < n$ for some $n\in \Nat$;
otherwise $x$ is {\em illimited}. There is no smallest illimited
hypernatural or hyperreal.

\item $x$ is {\em infinitesimal} if $|x| < 1/n$ for all $n\in\Nat$;
otherwise $x$ is {\em appreciable}.

\item If $x-y$ is infinitesimal, then $x$ and $y$ are {\em infinitely close},
or simply {\em near}. Notation: $x\near y$.

\item If $x$ is limited, then $x\near y$ for a unique $y \in \Real$.
$y$ is the {\em standard part} of $x$, denoted by $\Std{x}$ or $\st x$. 

\item If $x$ is standard, the collection of all hyperreals 
near $x$ makes up the {\em halo} (or {\em monad}) of $x$.
Halos of distinct reals are disjoint.

\item $x \gg y$ means $x > y$ but $x \not\near y$. $x \ll \infty$
is a synonym for ``$x$ is limited''.
\end{enumerate}
\end{defn}

Any formal statement $S$ we can make about $\Bool$, $\Nat$, and $\Real$
using Boolean operations
(conjunction $\wedge$, disjunction $\vee$ and negation $\neg$),
natural and real constants and arithmetic operations 
($0$,$1$,$+$,$\cdot$,$<$,$|\cdot|$) and 
bounded quantifiers ($\forall$ and $\exists$),
is turned into a statement $\Star{S}$ about $\Bool$, $\Star{\Nat}$ and $\Star{\Real}$ 
by putting stars on the constants and parameters. Leibniz' Principle,
which we require, says that $S$ and $\Star{S}$ have the same truth value.
For example, every natural is either even or odd, hence the same is true of a 
hypernatural, even if it is illimited. Every hypernatural has a successor.
There is a third hyperreal strictly between any two distinct hyperreals.
There is no real number squaring to $-1$, hence neither is there such a
hyperreal. Since we {\em identify} reals, naturals and Booleans with their 
$*$-transforms, `$\Star{(-1)}$' is simply `$-1$'.

\subsection{limits and continuity}
\label{sec:nsa-limits}

With individual functions added to our object language, we can discuss 
nonstandard characterizations of sequence limits and function continuity.
These are bread-and-butter applications of infinitesimal methods.
A real sequence $(x_n)$ is really a function $\Nat \to \Real$, and
has an extension $(\Star{x}_n)$ with $n$ running over $\Star{\Nat}$.
We want to show the equivalence of
``$x_n \to a$'' and  ``$\Star{x}_n \near a$ for $n \near \infty$''.

The formalization of $x_n \to a$ is
(${\Real}_{>0}$ abbreviates the reals greater than zero)
$\forall \epsilon \in {\Real}_{>0}\, \exists n\in\Nat.\,
m > n \Rightarrow |x_m - a| < \epsilon$.
Actually, we do not want to put $\epsilon$ and $n$ in the formal
statement, but keep them at the metalinguistic level.
Then we have simply, 
\begin{equation}
\forall m > n(\epsilon) \Rightarrow |x_m - a| < \epsilon,
\nonumber
\end{equation}
which is true in ${\Univ}$. 
Transfer says (recall, $n$ is the same as $\Star{n}$ and similarly for $\epsilon$)
\begin{equation}
\forall m > {n(\epsilon)} \Rightarrow |\Star{x}_m - {a}| < \epsilon.
\nonumber
\end{equation}
Now, if $m\near\infty$, then certainly $m > n(\epsilon)$.
And that's true no matter how small $\epsilon$ is, so $x_m \near a$.
For the other direction, suppose we know that $m\near\infty$ implies
$x_m \near a$. Then, given standard $\epsilon > 0$, the statement
\begin{equation}
\exists n\in\Star{\Nat}.\, m > n \Rightarrow |\Star{x}_m - a| < \epsilon
\nonumber
\end{equation}
is true in $\Star{\Univ}$; any illimited $n$ will serve.
Backward transfer gives us
\begin{equation}
\exists n\in{\Nat}.\, m > n \Rightarrow |{x}_m - a| < \epsilon.
\nonumber
\end{equation}
Since this procedure is valid for any $\epsilon \in \Real_{>0}$, 
we conclude that $x_n \to a$.

Continuity of a function $f: \Real \rightarrow \Real$ at $a$
has an infinitesimal characterization with a very similar flavor.
It is ``$x\near a$ implies $\Star{f}(x) \near \Star{f}(a)$'', 
and is demonstrated by a very similar technique.
Ordinarily, we might drop the stars on `$\Star{f}(x)$'
and `$\Star{f}(a)$', as there is neither ambiguity
nor possibility for confusion 
(Not even externally does `$a \near b$' make sense for $a,b\in\Real$.)

\subsection{metric, function, and topological spaces}
\label{sec:nsa-metric}

Any metric space $X$ is extended in much the same way as $\Real$.
For example, consider the space $\ell^2$ of square summable
sequences $x:\Nat \rightarrow \Real$, with the norm $\|x\|_2^2 = \sum_n |x_n|^2$.
Then $\Star{\ell^2}$ consists of sequences $\Star{\Nat} \rightarrow \Star{\Real}$.
$x\in \ell^2$ has an extension $\Star{x}\in\Star{\ell^2}$; it is a standard
element. Now, since $\sum |x_n|^2$ converges, for any standard $\epsilon > 0$,
there is $n(\epsilon)$ such that $\sum_{m\ge n(\epsilon)} |x_m|^2 < \epsilon$.
So using Transfer, we conclude that $\sum_{m \ge N} |x_m|^2 \near 0$ for
any illimited $N$. This is important to the nature of nearstandard elements
of $\Star{\ell^2}$. If $y\in\Star{\ell^2}$ has a standard part, then
$\sum_{n\in\Star{\Nat}} |(\st y)_n - y_n|^2 \near 0$.
This implies that $(\st y)_n \near \Std(y_n)$ for limited $n$,
as well as $\sum_{n\in\Nat} |\Std(y_n)|^2 \ll \infty$ and
$\sum_{n \ge N} |y_n|^2 \near 0$ for any illimited $N$.
We see, then, that there are basically two ways $y\in\Star{\ell^2}$ can fail to
be nearstandard (not nearstandard is also called {\it remote}).
First, it could have an illimited norm, $\|y\| \approx \infty$.
Second, it could have appreciable weight at infinity.
The first has as analog for $\Star{\Real}$, but the second does not.
This is the nonstandard manifestation of the fact that bounded subsets
of $\Real$ are relatively compact, but bounded subsets of $\ell^2$ are not.
The element $y_n = \delta_{nN}$ ($n$ illimited) is a remote element with
norm one.

If $d$ is a metric for $X$ in $\Univ$, then $\Star{d}$ is a
metric for $\Star{X}$. Each standard point $\Star{x} \in \Star{X}$
is surrounded by a halo (or `monad') of nonstandard points at infinitesimal 
distance from it. If $d(y,\Star{x}) \near 0$, then $x$ is the
{\em standard part} of $y$, written either $\Std{y}$ or $\st y$,
as convenient. $\Star{X}$ may well have points which are not
near any standard point. These are {\em remote}.
Indeed, if there are points in $X$ arbitrarily far from $x$
($\forall n\in\Nat\, \exists y \in X.\, d(y,x) > n$), then
there are points in $\Star{X}$ an illimited distance from
$\Star{x}$. These are necessarily remote, since any putative
standard part must be a finite distance from $x$. 
A function $f: X\to \Real$ is continuous at $x\in X$ if and
only if $\Star{f}$ maps the halo of $x$ into the halo of $f(x)$.
All these points are exemplified by the space $\ell^2$ of
the previous paragraph.

We touch briefly on
the generalization to an arbitrary topological space in $\Univ$.
There is only one essential difference. Without a metric, there
is no numerical measure of ``infinitely close''. Instead, the
halo of a standard point $x$ must be defined as $\cap \Star{U}$, 
where the intersection runs over all the standard open neighborhoods of $x$.
One can check that for a metric space this matches the definition we have
already given.

\subsection{$*$-transform and Leibniz' Principle}
\label{sec:nsa-Leibniz}

To do much interesting mathematics, we need higher-order types,
such as power sets and function spaces, that are so far lacking from the language.
If $A$ and $B$ are types in $\Univ$, then so is the product
type $A \times B$ consisting of all ordered pairs $(x,y)$ for $x\in A$ and 
$y\in B$, the sum (disjoint union) type $A + B$,
and the type $[A \rightarrow B]$ of all functions from $A$ to $B$.
With these type constructors, along with function abstraction and
application, we can fill out the universe $\Univ$
and extend the $*$-transform to a map 
${\Univ} \stackrel{*}{\rightarrow} \Star{\Univ}$
obeying the Transfer Principle.

\renewcommand{\labelenumi}{\theenumi}
\renewcommand{\theenumi}{\alph{enumi}.}
\begin{enumerate}

\item The range of $*$ is the {\em standard} objects.
These are {raw materials} from which the other,
{\em internal}, objects of $\Star{\Univ}$ are obtained
by ordinary mathematical operations ({\it de novo} function
space construction is not ordinary).

\item Individual Booleans, naturals and reals are urelements with
no internal structure; we identify them with their $*$-transforms.

\item If $x\in A$, then $\Star{x} \in \Star{A}$.
But if $A$ is infinite, $\Star{A}$ contains nonstandard
members as well (example: illimited hypernaturals).

\item $*$ respects product and sum construction,
$\Star{(A\times B)} = \Star{A}\times\Star{B}$, 
$\Star{(A + B)} = \Star{A} + \Star{B}$.
, tupling,
$\Star(x,y) = (\Star{x},\Star{y})$, and function application,
$\Star{(f(x))} = \Star{f}(\Star{x})$.

\item Only the function space constructor has an abnormal interpretation
in $\Star{\Univ}$. $\Star[A \rightarrow B]$ is only
a subset of $[\Star{A} \rightarrow \Star{B}]$, which we
write as $[\Star{A} \dashrightarrow \Star{B}]$.
Members of $[\Star{A} \dashrightarrow \Star{B}]$ are
{\em internal} functions.

\item
The power set constructor, ${\mathcal P}(\cdot)$, is basically syntactic 
sugar for $[\, \cdot \rightarrow \Bool]$. Other set operations are likewise 
abbreviations, which we use freely. 

\item $\Star{\mathbb U}$ consists of {\em internal} objects,
anything we can describe using the resources of the formal language 
and beginning with standard objects.
In particular, a standard set is internal, any {\em member} 
of a (member of a $\cdots$ of a) standard set is internal. 
If $c\in \Star{\Real}$ is infinitesimal,
then it is not standard (not of the form $\Star{b}$), 
but it is internal, and so is the subset $[0,c]$ of $\Star{\Real}$.

\item If $C$ and $D$ are internal subsets of $\Star{A}$ and $\Star{B}$,
respectively, $[C \dashrightarrow D]$ is obtained in the normal way
by selection and restriction.
Also, the internal function space constructor respects currying:
if $A$, $B$ and $C$ are internal, then 
$ {A} \dashrightarrow [{B} \dashrightarrow {C}] 
= ({A} \times {B}) \dashrightarrow {C} $.

\item The symbol `$\near$', the predicate `limited', the standard part
operation and the other elements of Definition \ref{def:nonstd} are 
{\em not} part of the object language.
Most especially, neither $\{x\in\Star{\Nat} : x \ll \infty\}$ nor
$\{x\in\Star{\Real} : x \near 0\}$ is internal. They are {\em external}. 
Because we can formally express the least-number principle 
(any nonempty subset of $\Nat$ has a least element), Transfer 
implies that it holds for {\em internal} subsets of $\Star{\Nat}$.
But, there is no smallest illimited hypernatural. Consistency thus
demands that these things be external sets. 
\end{enumerate}

\subsection{overspill}
\label{sec:nsa-overspill}

`Limited' and `infinitesimal' are the most important external concepts. 
We illustrate a typical pattern of use.
Suppose a sequence $x:\Star{\Nat}\rightarrow\Star{\Real}$ has the property that
$n x_n < 1$ for every limited $n$.
Then, the set $B = \{m\in\Star{\Nat} : \forall n < m, nx_n < 1 \}$ is evidently
internal. Therefore, $B$ cannot consist of only the limited
hypernaturals, since that is an external set. There must be some illimited $m$ in $B$.
Thus the property $nx_n < 1$ {\em spills over} into the illimited hypernaturals.
Similarly, if a formally expressible property holds for all appreciable hyperreals, 
it must hold for some infinitesimal ones.
Now suppose an internal sequence $(x_n)$ satisfies $x_n \near 0$ for all limited $n$. 
Can we use overspill? Since ``$\near$'' is external, it is not evidently applicable,
but there is a clever trick discovered by Robinson. The given condition 
implies $nx_n < 1$ for limited $n$. And for illimited $n$, $nx_n < 1$ 
implies $x_n\near 0$. Thus the property $x_n \near 0$ actually does
extend into the illimited hypernaturals.
This result, known as Robinson's Lemma, is used in \ref{sec:dichotomy}.


\providecommand{\newblock}{}

\end{document}